\documentclass[sigconf]{acmart}

\copyrightyear{2017}
\acmYear{2017}
\setcopyright{acmcopyright}
\acmConference{CCS '17}{October 30-November 3, 2017}{Dallas, TX, USA}
\acmPrice{15.00}
\acmDOI{10.1145/3133956.3134004}
\acmISBN{978-1-4503-4946-8/17/10}

\usepackage{booktabs} 
\usepackage[normalem]{ulem} 
\usepackage[caption=false]{subfig}

\newcommand{\Q}[0]{\text{Q}} 
\renewcommand{\P}[0]{\text{P}} 
\newcommand{\Popt}[0]{\P_{\texttt{opt}}}
\newcommand{\Qavg}[0]{\overline{\text{Q}}} 
\newcommand{\Qwc}[0]{\text{Q}^+} 
\newcommand{\Qwcmax}[0]{\text{Q}^+_{\texttt{max}}} 
\newcommand{\PAE}[0]{\P_{\texttt{AE}}} 
\newcommand{\PCE}[0]{\P_{\texttt{CE}}} 
\newcommand{\PWCCE}[0]{\P_{\texttt{WC-CE}}} 
\newcommand{\PWCAE}[0]{\P_{\texttt{WC-AE}}} 
\newcommand{\PGI}[0]{\P_{\texttt{GI}}} 
\newcommand{\dP}[1]{d_{P}(#1)} 
\newcommand{\dQ}[1]{d_{Q}(#1)} 
\newcommand{\Xalph}[0]{\mathcal{X}} 
\newcommand{\Zalph}[0]{\mathcal{Z}} 
\newcommand{\Xestalph}[0]{\hat{\mathcal{X}}} 
\newcommand{\RR}[0]{\mathbb{R}^2} 
\newcommand{\fdiscrete}[2]{p(#1|#2)}

\newcommand{\FQL}[1]{\mathcal{F}_{#1}}
\newcommand{\Fopt}[1]{\mathcal{F}^{\texttt{opt}}_{#1}}
\newcommand{\fcoin}[0]{f_{\texttt{coin}}}

\newcommand{\Coin}[0]{\texttt{Coin}}
\newcommand{\Lap}[0]{\texttt{Lap}}
\newcommand{\Cir}[0]{\texttt{Cir}}
\newcommand{\Gau}[0]{\texttt{Gau}}
\newcommand{\BA}[0]{\texttt{ExPost}}
\newcommand{\Exp}[0]{\texttt{Exp}}

\graphicspath{{./img/}}

\begin{document}

\title{Back to the Drawing Board: Revisiting the Design of Optimal Location Privacy-preserving Mechanisms}
\renewcommand{\shorttitle}{Revisiting the Design of Optimal Location Privacy-preserving Mechanisms}

\author{Simon Oya}
\affiliation{%
  \institution{University of Vigo}
  \streetaddress{}
  \city{} 
  \country{} 
  \postcode{}
}
\email{simonoya@gts.uvigo.es}

\author{Carmela Troncoso}
\affiliation{%
  \institution{IMDEA Software Institute}
  \streetaddress{}
  \city{} 
  \country{} 
  \postcode{}
}
\email{carmela.troncoso@imdea.org}

\author{Fernando P{\'e}rez-Gonz{\'a}lez}
\affiliation{%
  \institution{University of Vigo}
  \streetaddress{}
  \city{} 
  \country{}}
\email{fperez@gts.uvigo.es}

\begin{abstract}
In the last years we have witnessed the appearance of a variety of strategies to design optimal location privacy-preserving mechanisms, in terms of maximizing the adversary's expected error with respect to the users' whereabouts. In this work, we take a closer look at the defenses created by these strategies and show that, even though they are indeed optimal in terms of adversary's correctness, not all of them offer the same protection when looking at other dimensions of privacy. To avoid ``bad'' choices, we argue that the search for optimal mechanisms must be guided by complementary criteria. We provide two example auxiliary metrics that help in this regard: the conditional entropy, that captures an information-theoretic aspect of the problem; and the worst-case quality loss, that ensures that the output of the mechanism always provides a minimum utility to the users. We describe a new mechanism that maximizes the conditional entropy and is optimal in terms of average adversary error, and compare its performance with previously proposed optimal mechanisms using two real datasets. Our empirical results confirm that no mechanism fares well on every privacy criteria simultaneously, making apparent the need for considering multiple privacy dimensions to have a good understanding of the privacy protection a mechanism provides.
\end{abstract}



\maketitle

\section{Introduction}

Location based services raise important privacy concerns regarding the private information that exposing accurate location to service providers reveals~\cite{FreudigerSH12,GambsKC11,GolleP09,Krumm07,ZhengZXM09}. To protect users' privacy, the academic community has proposed a wide variety of location privacy-preserving mechanisms~\cite{BeresfordS04,GedikL05,GruteserG03,HohG05,KidoYS05,LuJY08,MeyerowitzC09,WangXHZLX12,YouPL07} that mostly work altering the users' actual location before exposing it to the service provider. The privacy evaluation of these proposals typically does not consider a strategic adversary, fostering an arms race in which defenses and attacks succeed each other without ever providing clear location privacy guarantees. To counter this effect, recent efforts focus on cutting the arms race short by either embedding the adversarial knowledge on the design process~\cite{Opt2012CCS,OptGeoInd2014CCS,Shokri15}, or providing guarantees independent of the adversary's prior~\cite{GeoInd2013CCS,OptGeoInd2014CCS,Shokri15}.
 
In this paper, we focus on sporadic user-centric protection mechanisms based on randomization, which preserve privacy by reporting a noisy version of the real location to the service provider according to a probability distribution. These mechanisms are adequate for applications that require infrequent location exposure, and can be run locally by the user.
In this scenario, approaches that embed the adversarial knowledge on the design process are based on a Bayesian modeling of the adversary~\cite{Quant2011SP}, and find optimal noise-generating mechanisms via linear optimization in which a target privacy objective is sought in presence of utility constraints~\cite{Opt2012CCS}. On the other hand, approaches that provide privacy guarantees independent of the adversary's prior are based on \textit{geo-indistinguishability}~\cite{GeoInd2013CCS}, an adaptation of differential privacy~\cite{Dwork06} to two-dimensional spaces, used by a number of works~\cite{FawazFS15,FawazS14,Ma2014}. Geo-indindistinguishability can be achieved optimally in terms of utility using expensive linear programming~\cite{OptGeoInd2014CCS}, or suboptimally using efficient remapping techniques that increase the utility of the query~\cite{Practical2016}. Finally, the Bayesian and the geo-indistinguishability approaches have been combined by Shokri~\cite{Shokri15} to obtain mechanisms  that guarantee geo-indindistinguishability while achieving a good performance against the Bayesian adversary.

Following the recommendation by Shokri et al.~\cite{Quant2011SP}, which has been taken as the standard by the community, all of these approaches use the adversary's \textit{correctness}, i.e., how close the adversary's estimate is to the correct answer, to evaluate location privacy. Usually, the adversary's correctness is measured as her expected estimation error, where this error is modeled using some distance metric between the real location and the adversary's estimation~\cite{ShokriFJH09}.  

In this paper, we aim at understanding the properties of the mechanisms output by these design strategies. We find that, when the target privacy notion is the adversary's expected estimation error, there are many optimal mechanisms that meet a desired quality loss constraint. While this may seem advantageous, we show that following such an optimization objective may result in the selection of naive mechanisms that obviously provide little privacy, e.g., alternating the exposure of the actual user location and a far away location. Indeed, this mechanism complies on average with the constraints of the problem. Yet, it results on little uncertainty for the adversary, effectively providing a false perception of privacy.

To counter such effect we argue that, depending on the user's preferences, the search for an optimal location privacy-preserving mechanism needs to consider more criteria than the error, contradicting the belief established by Shokri et al.~\cite{Quant2011SP}. As examples of complementary metrics to guide the design of protection mechanisms we propose the use of information-theoretic metrics, e.g., the conditional entropy, or a worst-case bound for quality loss. We provide efficient methods to construct mechanisms with respect to these criteria, and demonstrate that the remapping method introduced in~\cite{Practical2016} to improve the utility of geo-indistinguishability-based methods is in fact a straightforward generic scheme to build an optimal mechanism in terms of the expected estimation error from \textit{any} obfuscation mechanism.
We evaluate the effectiveness of the different mechanisms according to different privacy criteria using two real location datasets concluding that, generally, mechanisms that are optimal for one criterion do not necessarily perform well on others. 

To summarize, we make the following contributions:

\noindent\textbf{\checkmark} We provide a theoretical characterization of optimal location privacy-preserving mechanisms in terms of the mean adversarial error. We show that, for a given average quality loss, there is more than one optimal protection mechanism that maximizes the average privacy. This family of mechanisms forms a convex polytope in which different mechanisms provide different privacy guarantees.

\noindent\textbf{\checkmark} We demonstrate the limitations of evaluating defenses solely considering the correctness of the adversary~\cite{Quant2011SP}, and advocate for the use of complementary criteria to guide the design of location privacy-preserving mechanisms where the privacy guarantees provided are better understood.

\noindent\textbf{\checkmark} We provide algorithms to efficiently design mechanisms based on criteria other than the adversary's error. Furthermore, we demonstrate that remapping, previously proposed as an enhancement to geo-indistin\-guishability, is not only beneficial to improve the utility of this technique but can be used as a generic method to turn any obfuscation mechanism into optimal in terms of average adversarial error.

\noindent\textbf{\checkmark} We evaluate prior and new location privacy-preserving mechanisms on two real location datasets. Our results confirm that it is difficult to find optimal mechanisms that fare well on all criteria. This demonstrates that previous approaches to design location privacy-preserving mechanisms, while having solid foundations, oversimplify the design problem and generate defenses that overestimate the level of privacy offered to the user. 

This paper is organized as follows. In Section~\ref{sec:sysmodel}, we introduce our system model, and the quality loss and privacy metrics we consider in the paper. In Section~\ref{sec:PAEproblem} we study the consequences of choosing the average adversary error as the standard metric to evaluate location privacy, illustrating that mechanisms that are optimal by this criterion may provide little privacy. In Section~\ref{sec:complementary} we propose to consider auxiliary metrics to avoid bad mechanism choices in the optimization. As examples, we study the use of the conditional entropy and the worst-case quality loss. We evaluate several mechanisms built according to these new criteria in Section~\ref{sec:eval}, and offer our conclusions in Section~\ref{sec:conclusions}.

\section{System Model}
\label{sec:sysmodel}

We now describe our system model, which is in agreement with the framework for location privacy proposed by Shokri et al.~\cite{Quant2011SP}, and introduce the notation used throughout the paper, which is summarized in Table~\ref{tab:notation}. 

We consider a set of users that send queries with a geographical position of interest to a location based service to obtain a service (e.g., finding points of interest or nearby friends). The location of interest can be the current location of the user or some other location the user is interested in querying about. Users wish to obtain utility from the location based service, while keeping their whereabouts private from an adversary that can observe the locations in the queries, e.g, an eavesdropper of the user-server communication, or the service provider itself. 
In order to protect their locations, users employ a \emph{location privacy-preserving mechanism} that perturbs their location prior to exposing it to the server. We consider a strategic adversary that knows the protection mechanism operation, and has some knowledge about the users movement patterns. Given the observed perturbed location and her knowledge, the adversary tries to infer the user real location.

We model the locations queried by the users as a discrete set of \emph{points of interest} denoted by $\Xalph\doteq\{x_1,x_2,\cdots,x_N\}$. We refer to these locations as \textit{real} or \emph{input locations} since they are the actual locations that are input to the location privacy-preserving mechanism. We use $\pi(x)$ to denote the \emph{prior} probability that a user in the population queries the service provider about location $x$ ($\pi(x)\geq 0$ and $\sum_{x\in\Xalph} \pi(x)=1$). This prior can either represent the global behavior of all the users as in \cite{Practical2016}, or be tailored to a particular user, but we assume that it is known both by the user and the adversary and that it can be used to design the privacy-preserving mechanism. We also consider independence between queries, i.e., that the input locations $x$ from the same or other users are samples form i.i.d.~random variables given by $\pi$.

The set of possible locations reported by the location privacy-preserving mechanism is denoted by $\Zalph$. We assume that users can report any location in the world $\Zalph=\RR$. We refer to these locations as \emph{output locations}, as they are the outputs of the privacy-preserving mechanism. The mechanism itself is denoted by $f$ and modeled as a set of (continuous) conditional probability distributions, where $f(z|x)$ denotes the probability density function (pdf) of reporting the output location $z\in\RR$ when the real location of the user is $x\in\Xalph$ (note that $f(z|x)\geq 0$ and $\int_{\RR} f(z|x) dz=1$ for all $x\in\Xalph$). We represent discrete mechanisms, i.e., mechanisms with a discrete output domain, in $\RR$ with the Dirac delta function $\delta$. For example, the mechanism that maps any $x\in\Xalph$ to two particular outputs $z_1,z_2\in\RR$ with the same probability would be $f(z|x)=0.5\delta(z-z_1)+0.5\delta(z-z_2)$. For integration purposes, $\delta(z-z')$ must be understood as a two-dimensional Gaussian pdf centered at $z'$ whose variance is arbitrarily small.

When using a privacy-preserving mechanism $f$ to obtain privacy, the user experiences a loss on the quality of service due to the fact that she reports a location that might not be the location of interest, and may even be far away from this one. We use $\P(f,\pi)$ to denote the \emph{privacy} of the user, and $\Q(f,\pi)$ to denote her \emph{quality loss}. We specify particular instantiations of these functions below.

\begin{table}
\begin{center}
\caption{Summary of notation}
\label{tab:notation}
\begin{tabular}{p{.9 cm}l}
  \textbf{Symbol} & \multicolumn{1}{l}{\textbf{Meaning}}   \\
  \toprule
  $x$ & Input location the user is interested in querying about. \\
  $\Xalph$ & Set of valid input locations or points of interest. \\
  $z$ & Output location released by the mechanism, $z\in\RR$. \\
  $\hat{x}$ & Adversary's estimation of the input location, $\hat{x}\in\RR$. \\
  $\pi(x)$ & Prior probability that a user wants to query about $x$. \\
  $f(z|x)$ & Privacy mechanism. Pdf of $z\in\RR$ given $x\in\Xalph$. \\
  $f_Z(z)$ & Pdf of $z$, i.e., $f_Z(z)=\sum_{x\in\Xalph} \pi(x)\cdot f(z|x)$.\\
  $p(x|z)$ & Posterior probability of $x$ given $z$. \\
  \midrule
  $\dQ{x,z}$ & Quality loss distance function between $x$ and $z$. \\
  $\Qavg$ & Average quality loss metric, in \eqref{eq:Qavg}. \\
  $\Qwc$ & Worst-case quality loss metric, in \eqref{eq:Qwc}. \\  
  \midrule
  $\dP{x,\hat{x}}$ & Privacy distance function between $x$ and $\hat{x}$. \\
  $\PAE$ & Average error privacy metric, in \eqref{eq:PAE}. \\
  $\PCE$ & Conditional entropy privacy metric, in \eqref{eq:PCE} \\
  $\PGI$ & Geo-Indistinguishability privacy metric, in \eqref{eq:PGI}
 \end{tabular}
\end{center}
\end{table}

\subsection{Quality Loss Metrics}

We consider two possible definitions of quality loss: the average loss, and the worst-case loss. To this end we introduce $\dQ{x,z}$, a function that quantifies how much quality of service is lost by a user reporting output location $z$ when she is interested in input location $x$. Larger values of $\dQ{x,z}$ indicate a larger loss, and therefore a worse utility performance for the user. The canonical choice for this function is the Euclidean distance: $\dQ{x,z}=||x-z||_2$. Note that $\dQ{\cdot}$ does not need to be a metric in the mathematical sense: it could be any function that maps an input location and a released location to a loss value (e.g.,~a feeling-based utility metric as in~\cite{Bilogrevic2015predicting,Semantics2016PETS}).\\

\noindent{\textbf{Average Loss.}}
The average loss measures how much quality a user loses on average, and can be written as:
 \begin{equation} \label{eq:Qavg}
  \Qavg(f,\pi)=\sum_{x\in\Xalph} \int_{\RR}\pi(x) \cdot f(z|x) \cdot \dQ{x,z} dz\,.
 \end{equation}
This metric has been the typical choice of utility in the related literature \cite{Opt2012CCS,GeoInd2013CCS,OptGeoInd2014CCS,Elastic2015PETS,Practical2016} since it is very intuitive.  This metric also has the advantage of being linear with the mechanism $f$, which is very useful towards reducing the computational cost of mechanism design algorithms. Moreover, it makes the analysis of optimal algorithms in terms of average loss tractable.\\

\noindent{\textbf{Worst-case Loss.}}
Given a function that quantifies the point-wise loss as defined above, $\dQ{x,z}$, the worst-case loss is defined as:
 \begin{equation} \label{eq:Qwc}
  \Qwc(f,\pi)=\max_{\substack{x,z\\ \pi(x)>0\\ f(z|x)>0}} \dQ{x,z}\,.
 \end{equation}
The worst-case loss measures how much utility the user loses in the worst case possible. For example, if $\dQ{x,z}$ is the Euclidean distance and the user wants to query about $x$, a mechanism with $\Qwc(f,\pi)\leq 2\text{km}$ ensures that the output $z$ will not be further than 2km away from $x$. This property is very helpful for many applications that target nearby-type of services, since if the reported location is very far from the desired location then the result of the query would be generally useless for the user.

\subsection{Privacy Metrics}

We present now three notions of privacy: the average adversary error, the conditional entropy of the posterior distribution, and geo-indistinguishability.\\ 

\noindent{\textbf{Average Error.}}
The average error is the de-facto standard to measure location privacy since Shokri et al.~  \cite{Quant2011SP} argued that incorrectness determines the privacy of users. Consider that the adversary knows the prior $\pi$ and the mechanism $f$ chosen by the user. With this information, she produces an estimate $\hat{x}\in\Xestalph$ of the user's input location $x$. The choice of $\Xestalph$ depends on the computational power of the adversary. Since we assume that the user has the freedom to report any location in $\RR$, we also assume an unbounded adversary that can estimate locations on the whole world $\Xestalph=\RR$. Upon observing $z$, the adversary can build a \emph{posterior} probability mass function over the inputs, denoted as $p(x|z)$:
\begin{equation} \label{eq:posterior}
 p(x|z)=\frac{\pi(x)\cdot f(z|x)}{\sum_{x'\in\Xalph} \pi(x')\cdot f(z|x')}\,.
\end{equation}
Let $\dP{x,\hat{x}}$ be a function that quantifies the magnitude of the adversary's error when deciding that the input location was $\hat{x}$ when the input location is actually $x$. As in the case of the average loss $\Qavg$, this function $\dP{\cdot}$ does not necessarily need to be a metric (e.g., it can include the user sensitivity to an adversary learning semantic information such as in~\cite{Semantics2016PETS}). 
Given an output location $z$, the optimal decision for the adversary in terms of minimizing the average error is 
\begin{equation} \label{eq:advest}
 \hat{x}(z)=\underset{\hat{x}\in\RR}{\text{argmin }}\left\{ \sum_{x\in\Xalph} p(x|z) \cdot \dP{x,\hat{x}} \right\}\,.
\end{equation}

The average adversary's error, or just average error, is defined as the mean error incurred by an adversary that chooses the estimation $\hat{x}$ optimally given each observed $z$. Let $f_Z(z)=\sum_{x\in\Xalph} \pi(x) \cdot f(z|x)$ be the probability density function of $z$. Then, the average error is:
\begin{align}\label{eq:PAE}
  \PAE(f,\pi)&=\int_{\RR} f_Z(z) \sum_{x\in\Xalph} p(x|z) \cdot \dP{x,\hat{x}(z)} dz\\
             &=\int_{\RR} \min_{\hat{x}\in\RR}\left\{ \sum_{x\in\Xalph} \pi(x) \cdot f(z|x) \cdot \dP{x,\hat{x}} \right\} dz\,.
\end{align}
Note that mechanisms designed with $\PAE$ inherently protect against a strategic adversary, since the metric embeds the adversary's estimation. This metric has been used as part of the design objective in previous works \cite{Quant2011SP,Opt2012CCS}, and as a way of comparing the performance in terms of privacy of mechanisms designed with other different privacy goals in mind \cite{GeoInd2013CCS,OptGeoInd2014CCS,Elastic2015PETS,Practical2016}.\\

\noindent{\textbf{Conditional Entropy.}}
The conditional entropy is an information-theoretic metric that can be used to measure the adversary's uncertainty about the user's real location when $z$ is released. After observing $z$, the adversary builds the posterior $p(x|z)$ using \eqref{eq:posterior}. The uncertainty of the adversary regarding the value of $x$ given $z$ can be measured as the entropy of this posterior:
\begin{equation} \label{eq:condent0}
 H(x|z)\doteq-\sum_{x\in\Xalph}p(x|z)\cdot \log(p(x|z))\,.
\end{equation}
The conditional entropy measures the \emph{average} entropy of the posterior after $z$ is released. Formally,
\begin{equation} \label{eq:condent1}
 \PCE(f,\pi)=\int_{\RR} f_Z(z) \cdot H(x|z) dz\,,
\end{equation}
where $f_Z(z)$ is the probability density function of $z$, and $H(x|z)$ is a function of $z$ as defined in \eqref{eq:condent0}. Alternatively, using only the prior $\pi$ and the mechanism $f$, the conditional entropy can be written as
\begin{equation} \label{eq:PCE}
  \PCE(f,\pi)=-\sum_{x\in\Xalph} \int_{\RR} \pi(x) \cdot f(z|x)\cdot\log\left( \frac{\pi(x)\cdot f(z|x)}{\sum_{x'\in\Xalph} \pi(x')\cdot f(z|x')}\right)dz\,.
\end{equation}
Note that this metric does not depend on the geography of the problem, i.e., on the particular values of $x$ or $z$. If we use the base-two logarithm in the formula, then $\PCE$ can be interpreted as how many bits of information the adversary needs on average to completely identify $x$. This metric was disregarded as a possible privacy metric in \cite{Quant2011SP} due to being uncorrelated with the average error. In this work, we challenge such conclusion showing that considering solely the correctness of the adversary may lead to the design of mechanisms that offer low privacy. We show in Section~\ref{sec:complementary} how using the conditional entropy as a complementary privacy metric helps to avoid choosing those undesirable mechanisms.\\

\noindent{\textbf{Geo-Indistinguishability.}}
Geo-indistinguishability is an extension of the concept of differential privacy, originally a notion of privacy in databases, to the location privacy scenario. It was originally proposed in \cite{GeoInd2013CCS} and other works have continued the research on this line~\cite{OptGeoInd2014CCS,Elastic2015PETS,Practical2016}. Formally, $\epsilon$-geo-indistinguishability requires the following condition to be fulfilled by a location privacy-preserving mechanism $f$,
\begin{equation} \label{eq:geoindcondition}
 \int_{A} f(z|x) dz\leq e^{\epsilon\cdot \dP{x,x'}} \cdot \int_{A} f(z|x') dz \,,\quad\forall x,x'\in\Xalph\,,\forall A\subseteq\RR\,.
\end{equation}
This requirement ensures that given an area $A\subseteq\RR$, the probability of reporting a point $z$ in that area if the original location was $x$ over any other location $x'$ within some distance around $x$, is similar, and therefore $x$ and $x'$ have some degree of statistical indistinguishability. In this definition, $\dP{x,x'}$ is a function that quantifies how indistinguishable $x$ and $x'$ are: smaller values of $\dP{x,x'}$ indicate a higher indistinguishability, as the constraint becomes tighter. The privacy parameter in this definition is $\epsilon$: larger values of $\epsilon$ indicate a looser constraint that allows $f(z|x)$ and $f(z|x')$ to be more different, and therefore $x$ and $x'$ become more distinguishable. Smaller values of $\epsilon$ force the probability density functions $f(z|x)$ and $f(z|x')$ to be closer, providing more privacy. Note that, if for a single input location $x$ there is a positive probability of reporting the output in a region $A\subseteq\RR$, $\int_{A} f(z|x) dz>0$, then that must also be true for every other input location $x'$. Also, note that geo-indistinguishability is independent of the prior $\pi$.

The typical choice of $\dP{x,x'}$ in geo-indistinguishability is the Euclidean distance \cite{GeoInd2013CCS,OptGeoInd2014CCS}. Many geo-indistinguishability mechanisms rely on the fact that $\dP{x,x'}$ is a metric (specifically, in the fact that it satisfies the triangular inequality $\dP{x,x'}\leq \dP{x,z}+\dP{x',z}$) to prove that they meet the condition in \eqref{eq:geoindcondition}.

Although geo-indistinguishability is generally considered a privacy guarantee and not itself a metric, we can adapt it to represent an equivalent concept to our generic metric $\P(f,\pi)$. Given a mechanism that provides $\epsilon$-geo-indistinguishability, it is straightforward to see that it is also $\epsilon'$-geo-indistinguishable if $\epsilon'>\epsilon$. Since a smaller $\epsilon$ denotes more privacy, it makes sense to define the geo-indistinguishability level provided by a mechanism $f$ according to the smallest $\epsilon$ it guarantees. Also, since we are defining $\P(f,\pi)$ as a magnitude that grows with the protection of the users, we choose to define our measure of geo-indistinguishability, $\PGI(f)$, as the inverse of the smallest $\epsilon$ guaranteed by the mechanism. Given the mechanism $f$, we write
\begin{equation} \label{eq:PGI}
 \PGI(f)=\inf_{\substack{x,x'\in\Xalph \\ z\in\RR}} \dP{x,x'}\cdot\left|\log\frac{f(z|x)}{f(z|x')}\right|^{-1}\,,
\end{equation}
where we assume by convention that $\log(\frac{0}{0})=0$ and that $\dP{x,x'}=||x-x'||_2$ is the Euclidean distance. Larger values of $\PGI$ indicate more privacy, and the mechanism guarantees $1/\PGI$-geo-indistingui\-shability.

\section{Limitations of the Expected Adversary Error Based Evaluation}
\label{sec:PAEproblem}

The most standard way to assess the location privacy provided by two mechanisms has been the evaluation of the trade-off between their average adversary error $\PAE$ and their average loss $\Qavg$. The use of the average error as yardstick for location privacy was proposed in \cite{Quant2011SP} under the general notion of correctness, and its use as a way of comparing mechanisms was followed by many of the subsequent works \cite{Opt2012CCS,GeoInd2013CCS,OptGeoInd2014CCS,Elastic2015PETS,Semantics2016PETS,Practical2016}. The choice of distance functions $\dP{\cdot}$ and $\dQ{\cdot}$ for both the average error and the average loss in these works is mostly the Euclidean distance \cite{Opt2012CCS,GeoInd2013CCS,OptGeoInd2014CCS,Semantics2016PETS,Practical2016} although some of them also consider the Hamming distance~\cite{Quant2011SP,Opt2012CCS,OptGeoInd2014CCS} or semantic distances for privacy~\cite{Elastic2015PETS,Semantics2016PETS}.

In this section, we show the problems that stem from this established \emph{2-dimensional} evaluation approach. We start by studying the properties of mechanisms that are optimal according to these two metrics. Then, we introduce a new mechanism that we call the \emph{coin mechanism}, and use it as an example that brings to light the flaws of judging the privacy of a mechanism by its performance in terms of average error and average loss.

\subsection{Study of the Established Mechanism Evaluation}

We start our analysis by assuming that the choice of distance functions $\dP{\cdot}$ and $\dQ{\cdot}$ is the same for simplicity, which is a typical choice in related works (e.g., both are the Euclidean distance). We denote this by $\dP{\cdot}\equiv\dQ{\cdot}$. At the end of the section, we argue what happens when this is not the case. We also introduce two definitions. First, let $\FQL{\Q}$ be the set of all the mechanisms that achieve an average loss smaller or equal than $\Q$. Formally,
\begin{equation}
 \FQL{\Q}\doteq\left\{ f\, | \, \Qavg(f,\pi)\leq\Q \right\}.
\end{equation}
Also, let $\Fopt{\Q}\subseteq\FQL{\Q}$ be the set of all mechanisms $f\in\FQL{\Q}$ that are optimal in terms of average adversary error, i.e.,
\begin{equation}
 \Fopt{\Q}\doteq\left\{ f\, |\, f\in\FQL{\Q}\,,\enskip \PAE(f,\pi)\geq\PAE(f',\pi)\enskip\forall f'\in\FQL{\Q} \right\}\,.
\end{equation}
We call a mechanism inside $\Fopt{\Q}$ optimal, since it achieves as much privacy as possible among all the mechanisms with the same quality loss. 
We state the following lemma:
\begin{lemma}
 The set of optimal mechanisms with respect to the average privacy $\PAE$ and the average loss $\Qavg$ is a convex polytope.
\end{lemma}
\begin{proof}
Let the privacy achieved by any mechanism in $\Fopt{\Q}$ be $\Popt(\Q)$. Then, we can define this set as
\begin{equation} \label{eq:FoptPoly}
 \Fopt{\Q}=\{f\, |\, \PAE(f,\pi)=\Popt(\Q),\quad \Qavg(f,\pi)\leq \Q\}\,,
\end{equation}
and since $\PAE(f,\pi)$ and $\Qavg(f,\pi)$ are linear operations with $f$, \eqref{eq:FoptPoly} can be written as an intersection of half-spaces, which forms a convex polytope.
\end{proof}
Note that the proof also applies to the case where $\dP{\cdot}\nequiv\dQ{\cdot}$ (e.g., privacy as the average Hamming error of the adversary and quality loss as the average Manhattan distance). The same outcome can be derived for the conditional entropy and geo-indistinguishability, although we leave those results out of the scope of this work.

This lemma shows that there is a \emph{family} of optimal mechanisms that lie inside a convex polytope, instead of just a single mechanism. All of them provide the same (maximal) privacy for the same quality loss constraint so, in principle, they are equally useful. In what follows, we show why this is not the case.

We start by introducing the concept of remapping. A remapping $g$ is a function $g:\RR\to\RR$ that maps an output $z\in\RR$ to another output $z'\in\RR$ according to the probability density function $g(z'|z)$. It is well known that if we generate a mechanism $f'=f\circ g=\int_{\RR} g(z'|z) \cdot f(z|x) dz$, then the privacy of $f'$ in terms of average error, conditional entropy or geo-indistinguishability is not smaller than that of $f$. This is reasonable, as the remapping $g$ is independent from $x$, and thus it does not reveal any information about it. The optimal Bayesian remapping is defined as follows:
\begin{definition}[Optimal remapping]
 Given a mechanism $f$, its optimal remapping is the one that minimizes the average loss of the composition $f'=f\circ g$, i.e., $g(z'|z)=\delta(z'-r(z))$, where
 \begin{equation} \label{eq:optremapping}
  r(z)=\underset{z'\in\RR}{\text{argmin }} \sum_{x\in\Xalph} \pi(x)\cdot f(z|x)\cdot \dQ{x,z'}\,.
 \end{equation}
\end{definition}
This remapping assigns each location $z$ to the location $r(z)$ in \eqref{eq:optremapping}, and is used in \cite{Practical2016} as a way of improving the utility of geo-indistinguishability mechanisms. Now, we show that it can also be used not only to reduce the quality loss of mechanisms but to achieve optimal mechanisms in terms of average error privacy:
\begin{theorem} \label{theo:optmech}
 Let $g$ be an optimal remapping for mechanism $f$, and let $f'$ be the composition $f'=f\circ g$. If $\dP{\cdot}\equiv\dQ{\cdot}$, then $f'$ is an optimal mechanism, i.e., $f'\in\Fopt{\Qavg(f',\pi)}$.
\end{theorem}
The proof is provided in the Appendix.

This theorem provides a straightforward way of building an optimal mechanism $f'$ from any mechanism $f$. The idea is to reassign each output $z$ of $f$ to another symbol $z'$ such that the average quality loss is minimized. Doing this for every output ensures that the quality loss cannot be further reduced, and since the distance function used to evaluate quality loss and privacy is the same, the best estimation the adversary can do of $x$ is just to keep the released value. Note that the $\Qavg(f',\pi)\leq\Qavg(f,\pi)$. This means that, in order to find an optimal mechanism $f'$ for a target quality loss $\Qavg(f',\pi)=\Q$ using the remapping strategy, one has to adjust the loss of the mechanism $f$ (e.g., by tuning its variance if it is a noise mechanism) until $f'$ achieves the desired average loss $\Q$.

It is straightforward to see that, if the optimal remapping for a mechanism $f$ is just doing nothing, then it means $f$ is optimal:
\begin{corollary} \label{coro:optmech}
 If the optimal remapping in \eqref{eq:optremapping} for a mechanism $f$ is $g(z'|z)=\delta(z'-z)$, then $f$ is optimal for its quality loss $\Q$, i.e., $f\in\Fopt{\Q}$.
\end{corollary}
This is a very convenient way of proving the optimality of a mechanism when $\dP{\cdot}\equiv\dQ{\cdot}$. Another way of seeing that such mechanism is optimal, is by realizing that with this choice of metrics, the privacy is upper bounded by the quality loss $\PAE(f,\pi)\leq\Qavg(f,\pi)$, and the upper bound is indeed achieved when an optimal mechanism is used. We note that the fact that $\PAE(f,\pi)=\Qavg(f,\pi)$ for optimal mechanisms is not new, as it was already mentioned in \cite{GeoInd2013CCS} about the mechanisms in \cite{Opt2012CCS}.

\subsection{The Coin Mechanism and the Flaws of the Traditional Approach}
\label{sec:coin}

We now discuss the following mechanism, which we call \emph{the coin mechanism}, and prove that it is optimal. Let $z^*$ be the output location that minimizes the average quality loss of a mechanism that always reports that location regardless of the input $x$. Formally,
\begin{equation} \label{eq:optoutputcoin}
 z^*\doteq \underset{z\in\RR}{\text{argmin }} \sum_{x\in\Xalph} \pi(x)\cdot \dQ{x,z}\,.
\end{equation}
As an example, if we measure the point-to-point loss as the mean squared error $\dQ{x,z}=||x-z||_2^2$, then $z^*$ will be given by the mean $z^*=\sum_{x\in\Xalph} \pi(x) \cdot x$. If the loss is measured as the Euclidean distance, then $z^*$ is the geometric median of $\pi$. Given a generic distance function $\dQ{\cdot}$, the optimal output location $z^*$ can be computed by solving the optimization problem in \eqref{eq:optoutputcoin}. 

Let $\Q^*$ be the average quality loss achieved by a mechanism that always reports $z^*$ regardless of the input. We construct the following mechanism, which we denote $\fcoin$. First, we fix a desired quality loss $\Q\leq \Q^*$ and compute $\alpha\doteq1-\Q/\Q^*$. Then, we build
\begin{equation} \label{eq:coin}
 \fcoin(z|x)=\alpha\cdot\delta(z-x)+(1-\alpha)\cdot\delta(z-z^*)\,,
\end{equation}
where $z^*$ is in \eqref{eq:optoutputcoin}. This mechanism can be easily explained and implemented simulating a coin flip. We first set our desired quality loss $\Q\leq\Q^*$. Note that it would not make sense to fix $\Q$ to a value larger than $\Q^*$ since we would not achieve more privacy by doing so; a mechanism that always reports $z^*$ and has an average loss of $\Q^*$ yields the highest privacy allowed by $\pi$. Then, we compute $\alpha=1-\Q/\Q^*$ and set it as the probability that our coin shows heads. Assume we are interested in querying about a location $x\in\Xalph$, so we flip the coin. If the coin shows heads, then we report our desired location $z=x$. If the coin hits tails, then we report $z^*$ regardless of the value of $x$. It is easy to see that the average loss of \eqref{eq:coin} is indeed $\Q$, by the linearity of this metric with $f$.

\begin{proposition}
 The coin mechanism obtained for quality loss $\Q$ achieves the maximum average adversarial error possible given a constraint on the average quality loss, i.e., $\fcoin(\Q)\in\Fopt{\Q}$, if both are measured with the same distance function $\dP{\cdot}\equiv\dQ{\cdot}$.
\end{proposition}
The proof is straightforward using the result in Corollary~\ref{coro:optmech}.

We now reason why, even though the coin mechanism is optimal by the standards that have been used to evaluate privacy in prior works (i.e., $\PAE$ and $\Qavg$), this mechanism is hardly desirable for any user. When the coin shows heads, the adversary observes $z$. If $z\neq z^*$, the adversary knows for sure that the user was interested in querying about $x=z$ and therefore the user has no privacy at all. In this case, for privacy issues, there was no point in using the mechanism. When the coin shows tails, the user is mapped far away to $z^*$. The adversary observes $z^*$ and has no idea where the user is, besides the prior $\pi$ that was already known by her. In this case, the privacy of the user is maximal, but the quality loss is very large, since $z^*$ is almost always very far away from the user. The quality loss is so large that the utility the user gets from this realization of the mechanism can be considered zero, so we can say that there was no point in using the mechanism in this case either. We have reached the issue we mentioned earlier: there is a mechanism, optimal by classic location privacy standards~\cite{Quant2011SP}, that is useless both from the privacy and the quality loss point of view. This shows that there is a fundamental problem with the classic way that has been used to evaluate location privacy mechanisms.

\subsection{The reach of this problem}
One could think that the problem of this bi-dimensional evaluation approach lies on the fact that one cannot use the same metric to measure quality loss and privacy, e.g., the Euclidean distance. However, even with different metrics, mechanisms similar to the coin can be derived. For example, if privacy is the average mean squared error and quality loss is measured as the average Manhattan distance (i.e., the $l_1$ norm), a deterministic mechanism that consists on reporting the real location on most of the places and mapping to the other side of the Earth in some others is optimal, due to the fact that the MSE grows quadratically with the distance, while the $l_1$ (or any $l_p$ norm) does not. In our evaluation, we show an example where a mechanism optimized for $\PAE$ and $\Qavg$ with a different pair of distance functions $\dP{\cdot}\neq\dQ{\cdot}$ suffers from the coin issue. The problem does not arise from the particular distance functions $\dP{\cdot}$ and $\dQ{\cdot}$ one uses to evaluate the average error and loss, but from the fact that these metrics are \emph{averages}, and as such they do not restrict the minimum privacy of a single use of the mechanism or the maximum quality loss of the mechanism, they just ensure that the average is good. We believe that, while evaluating the average behavior of a mechanism is not an erroneous notion per-se, it must be handled with care to avoid undesirable results, such as the coin mechanism.

As a concluding remark, we would like to note that we have shown this problem assuming that the outputs of the mechanism and the values estimated by the adversary are points in $\RR$, for notational simplicity and generality. An important fraction of previous works \cite{Quant2011SP,Opt2012CCS,OptGeoInd2014CCS,Elastic2015PETS,Semantics2016PETS} assume a discrete model where the set of output values $\Zalph$ and estimated values $\Xestalph$ are the centers of a grid over the map or points of interest such as $\Xalph$. In these scenarios, one can derive a similar mechanism, where hitting tails means that the user reports the location out of the allowed ones that minimizes the average error. That mechanism can also be shown to be optimal in terms of average error and loss, although it is not a desirable mechanism for any user. For completeness, we also evaluate this scenario in our experiments. The same applies to the case where instead of having discrete input locations $\Xalph$, users can report any point in $\RR$ (for example, a tracking or a date finder application). The coin mechanism in \eqref{eq:coin} can be applied directly to this scenario, and it can be shown to be optimal (changing the summations over $\Xalph$ to integrals). It is clear that using the traditional evaluation approach has flaws in all these scenarios and we must find a solution to this.

\section{Complementary Mechanism Evaluation Criteria}
\label{sec:complementary}

So far we have seen that evaluating mechanisms based solely on the average error and quality loss does not reflect whether a mechanism is actually more beneficial than another one, due to the fact that some undesirable mechanisms are deemed optimal by this approach. In this section, we propose a solution to this evaluation procedure that consists in incorporating complementary evaluation criteria that add different perspectives to the performance of a mechanism in terms of privacy and quality loss.

We propose two metrics, that are not intended to be used as a replacement of the average error and average loss but in combination with them, adding new dimensions to the privacy vs.~quality loss trade-off. The first metric we propose is the conditional entropy, a privacy metric that helps detecting inconsistent mechanisms such as the coin. The second one is the worst-case loss, a quality loss metric that provides a way of staying out of mechanisms that might yield no utility for the user at all. We comment on the implementation of mechanisms that take these metrics into consideration, and propose a mechanism that maximizes the conditional entropy while being optimal in terms of average error and quality loss. We finish the section describing other alternative privacy metrics.

\subsection{The Conditional Entropy as a Complementary Metric}

\subsubsection{Usefulness of the Conditional Entropy}

One of the problems of the coin mechanism can be seen from an information-theoretic point of view. The coin is a binary mechanism, in the sense that each input location can only be mapped to itself or to a fixed point in the map. From the adversary's perspective, this means that if the coin shows heads the adversary has no uncertainty at all about the user's input location, and if it shows tails the uncertainty is maximal. The conditional entropy can be used to detect these scenarios where the adversary has no uncertainty about $x$. Recalling \eqref{eq:condent1}, the conditional entropy can be written as
\begin{equation} \label{eq:condent2}
 \PCE(f,\pi)=\int_{\RR} f_Z(z) \cdot H(x|z) dz\,,
\end{equation}
where $H(x|z)\doteq-\sum_{x\in\Xalph} p(x|z)\cdot \log(p(x|z))$ is the entropy of the posterior after a location $z$ is released. It is clear that \eqref{eq:condent2} is an average over the entropy of all the posteriors. However, contrary to the average error, the conditional entropy is an average over functions $H(z|x)$ that are strictly concave with $f$. This means that in order to perform well in terms of the conditional entropy, a mechanism must spread its uncertainty among every posterior $p(x|z)$ instead of achieving maximal uncertainty with some outputs and zero uncertainty with others, as the coin does.

Another interesting property of the entropy is that it is not a geographical metric. The entropy of a posterior $H(x|z)$ does not depend on the coordinates of the input locations or the semantic information tied to them (e.g., if the location is a hospital or a club). The entropy only depends on how evenly the posterior is distributed among the input locations. This probabilistic aspect of privacy, defined as \emph{uncertainty} in \cite{Quant2011SP}, cannot be captured by other privacy notions such as correctness (e.g., the average adversary error). Due to the geographic nature of the location privacy problem, we cannot judge a mechanism based solely on its entropy. However, using it as an additional dimension of privacy gives a more complete picture of the performance of a mechanism.

We would like to point out that this notion of uncertainty provided by the entropy was disregarded as a reasonable privacy metric in \cite{Quant2011SP} based on the fact that, since it is not correlated with the adversary error, it does not capture how hard is for the adversary to estimate the real input location. We claim that it is indeed the fact that the entropy is not correlated with the adversary error which gives it a special value as a \emph{complementary} metric of privacy. The same way that semantic location privacy metrics have been proposed together with geographic metrics~\cite{Elastic2015PETS,Semantics2016PETS} to give different perspectives on the problem, the conditional entropy is a tool that gives valuable information about the protection provided by the mechanism not captured by the average error.

We would like to make two remarks regarding the entropy. First, the conditional entropy $\PCE(f,\pi)$ must be taken into account together with the mutual information $I(X;Z)$ to get a full picture of the information-theoretic properties of the mechanism. The conditional entropy represents the average amount of uncertainty the adversary has about the real location $x$ after observing $z$. A small value of conditional entropy indicates low uncertainty, and therefore we might get the impression that a mechanism with such small value provides low privacy. However, it might have been possible that the entropy of the prior was already low, and therefore even if the mechanism was perfect from the privacy point of view (i.e. it did not reveal any information, $I(X;Z)=0$), there is nothing any mechanism could have done to avoid having a low conditional entropy. We must therefore take into account the mutual information or, equivalently, the entropy of the prior $\pi$, when interpreting the value given by the conditional entropy.

The second remark is that the conditional entropy must not be tailored to a particular adversary with a possibly wrong knowledge of the prior $\pi$. In this work, we have assumed that the prior $\pi$ models the choice of input locations by the users, and therefore the correct way of computing the entropy is by using $\pi$ in the formulas above. This entropy must be regarded as the uncertainty that a very strong passive adversary with full knowledge of the behavior of the users would have when observing $z$.

\subsubsection{Implementation of Mechanisms with large Conditional Entropy}
\label{sec:BlahutArimoto}
We now look for a mechanism that is optimal in terms of the average error and average loss, i.e., a mechanism in $\Fopt{\Q}$, that also achieves as much conditional entropy as possible. This problem is equivalent to the rate-distortion problem~\cite{cover2012elements} of finding a pdf $f(z|x)$ that minimizes the mutual information between $x$ and $z$ subject to a quality loss constraint, which can be solved iteratively by implementing the Blahut-Arimoto algorithm.
For this, we must first restrict our output to a discrete alphabet $\Zalph$ for computational reasons. The more points we assign to this alphabet and the more evenly we cover the space where we want to compute the mechanism with them, the better its performance will be. Since both the input and output domains are discrete, the mechanism is determined by the probabilities of reporting $z$ when the user is in $x$, that we denote by $\fdiscrete{z}{x}$ here for clarity. We start with an initial mechanism, for example uniform mapping $\fdiscrete{z}{x}=1/|\Zalph|$. Then, we perform the following steps:
\begin{enumerate}
  \item We compute the probability mass function of each the output:
        \begin{equation}
         P_Z(z)=\sum_{x\in\Xalph} \pi(x)\cdot \fdiscrete{z}{x}\,,\qquad\forall z\in\Zalph\,.
        \end{equation}
  \item We update the mechanism as follows:
        \begin{equation} \label{eq:expoBA}
         \fdiscrete{z}{x}=P_Z(z)\cdot e^{-b \cdot \dQ{x,z}}\,,\qquad\forall x\in\Xalph,z\in\Zalph.
        \end{equation}
  \item We normalize the mechanism:
	\begin{equation}
         \fdiscrete{z}{x}=\frac{\fdiscrete{z}{x}}{\sum_{z'\in\Zalph} \fdiscrete{z'}{x}}\,,\qquad\forall x\in\Xalph,z\in\Zalph.
        \end{equation}
        We skip this step for the outputs $z$ with $P_Z(z)=0$.
  \item We repeat these steps until the change in the probabilities $\fdiscrete{z}{x}$ is below some threshold.
\end{enumerate}
The value of $b$ in the second step needs to be tuned to change the quality loss of the mechanism $\Qavg(f,\pi)$ and cannot be pre-computed to achieve an exact value of average loss. Larger values of $b$ yield mechanisms with less quality loss, and therefore less average error privacy and less conditional entropy. Finally, we obtain our mechanism $f(z|x)$ by applying the optimal remapping to the discrete mechanism defined in $\Xalph\to\Zalph$ by the probabilities $\fdiscrete{z}{x}$. This ensures that the resulting mechanism is optimal from the adversary error privacy point of view.

We make two remarks regarding this algorithm. The first one is about its computational cost. The operations in the three steps above are not expensive as they only include multiplications and additions. The number of elements we need to compute in order to build $\fdiscrete{z}{x}$ is $N\doteq|\Xalph|\cdot|\Zalph|$. The first step above consists of $N$ products and additions. In the second step $e^{-b\cdot\dQ{x,z}}$ can be precomputed as $b$, $\Xalph$ and $\Zalph$ do not change during the algorithm, so we only have to make $N$ multiplications, and in the third step we compute $|\Xalph|$ values of $\sum_{z'\in\Zalph}\fdiscrete{z'}{x}$ and then perform $N$ divisions. It is clear then that the cost grows with the sizes of $\Xalph$ and $\Zalph$. However, the algorithm only needs to be computed once for all the users, which can be done in the cloud, and even if the prior $\pi$ varies we can use a previously computed algorithm as initialization of the iteration above to get a fast update of the mechanism.

The second remark is that the mechanism produced by this algorithm also satisfies $2b$-geo-indistinguishability (the proof is in the Appendix). This is a byproduct property that was not part of the reasoning behind the algorithm and it does not imply that the conditional entropy and geo-indistinguishability are related. In fact, these are \emph{fundamentally different} notions: the former is an average metric that only considers the probabilistic (and not the geographic) aspect of the problem, while the latter is a worst-case metric that also considers the geography of the problem. Also, if we truncate the optimal conditional entropy mechanism, we obtain a mechanism that is almost optimal in terms of conditional entropy but does not provide \emph{any} level of geo-indistinguishability.

We evaluate this mechanism and others with respect to the conditional entropy and the traditional metrics in Section~\ref{sec:eval}.

\subsection{The Worst-Case Quality Loss as a Complementary Metric}
\label{sec:Qwc}

\subsubsection{Usefulness of the Worst-Case Quality Loss}
After analyzing the privacy problems of the coin mechanism, we now turn to the utility point of view. The great drawback of the coin mechanism from the quality loss perspective is that if the coin shows tails then the server's response to the user's query will most likely be useless due to the great quality loss incurred by reporting $z^*$. We can think of many applications where, if the Euclidean distance between $x$ and $z$ is larger than a certain value, the user gets literally nothing from the server response. For example, if we are close to a point of interest $x$ and we want to find a nearby hospital, querying about a location $z$ in another city will likely return a useless response from the server. In that case, we could think of generating another output and query the server again because we did not get what we were hoping for. By doing so, the privacy properties of the mechanism change, and in the case of the coin it is equivalent to always revealing our true location.

A solution to this utility issue consists in imposing a worst-case quality loss constraint on the mechanism, i.e.,
\begin{equation} \label{eq:Qwcconst}
 \Qwc(f,\pi)=\max_{\substack{x,z\\ \pi(x)>0\\ f(z|x)>0}} \dQ{x,z}\leq\Qwcmax\,.
\end{equation}
To put it simply, we want a mechanism that releases output locations within $\Qwcmax$ from the input location, i.e., a \emph{bounded mechanism}. The upper bound $\Qwcmax$ would be tuned depending on the application in question, so that a user never gets a worthless result. When used together with the average error and the average loss, the worst-case loss metric reveals those mechanisms we might want to avoid using. It is easy to see that the coin mechanism, although optimal in terms of $\PAE$ and $\Qavg$, gives a very large value of $\Qwc(\fcoin,\pi)$, which manifests its uselessness.

An interesting consequence of setting a maximum worst-case quality loss constraint when designing a mechanism is that it can simplify the computational cost of the protocol that implements or computes it. For example, take the case of the works in \cite{Opt2012CCS,OptGeoInd2014CCS}, where authors assume a discrete set of output locations $\Zalph$ and propose to solve a linear program to find an optimal mechanism (in terms of average error and geo-indistinguishability, respectively). The constraint in \eqref{eq:Qwcconst} reduces the amount of variables that need to be computed in these programs (only a subset of $\Zalph$ are possible outputs for each input $x\in\Xalph$), as well as the amount of constraints, which in turn decreases drastically the computational cost of the problem. In other implementations of mechanisms, where $f$ is not explicitly derived but computed by adding (continuous) noise and then computing a remapping using the posterior (c.f.~\cite{Practical2016}), having a worst-case quality loss constraint reduces the amount of inputs that need to be considered when computing the posterior, effectively reducing the computational cost of the algorithm.

Finally, we would like to note that this metric exposes a basic problem with geo-indistinguishability mechanisms. As mentioned before, when using a geo-indistinguishability mechanism, if a user with input location $x$ has non-zero probability of reporting $z\in A\subseteq\RR$, then when the input location is any other $x'\in\Xalph$ she must assign a non-zero probability to reporting $z\in A$. This means that for any geo-indistinguishable mechanism $f$, the worst-case quality loss metric $\Qwc(f,\pi)$ gives a huge value and the probability of getting a useless response from the server would be larger than zero. One could argue that, given the nature of the geo-indistinguishability guarantee, the probability of reporting a location $z$ far from $x$ is low and decreases exponentially with the distance between them, so we could disregard such an event from happening. However, if we really truncate the mechanism to ensure that the probability of going very far is zero, then the mechanism does not provide any geo-indistinguishability guarantee at all. It is then clear that geo-indistinguishability mechanisms are problematic from the quality loss point of view, and if a user gets zero utility from a realization of the mechanism she cannot re-use it immediately, otherwise the privacy guarantee is violated. We comment on a possible solution to this problem below.

\subsubsection{Implementation of Mechanisms with Worst-Case Quality Loss Constraint}
Now we set the task of designing a mechanism that achieves a good value of worst-case quality loss or, alternatively, that ensures that the worst-case quality loss is below some bound $\Qwc(f,\pi)\leq\Qwcmax$. The straightforward approach, given a mechanism $f$, is to truncate the mechanism (for example, by generating samples of $z$ until one of them ensures that $\dQ{x,z}\leq\Qwcmax$, and then releasing that $z$). This approach is reasonable, but one must take into account that the privacy properties of this new truncated mechanism $f'$ are not the same as the original mechanism $f$, and therefore they must be re-evaluated.

Another issue that concerns the design of bounded mechanisms is that a deterministic remapping \eqref{eq:optremapping} might violate a $\Qwc$ constraint (i.e., even if $f$ guarantees the $\Qwc$ constraint, a composition $f'=f\circ g$ might not guarantee it). Finding a bounded mechanism that achieves as much privacy as an unbounded one in $\Fopt{\Q}$ can be an impossible task, due to the fact that the polytope defined by $\Qwc(f,\pi)\leq\Qwcmax$ might be disjoint with $\Fopt{\Q}$. However, we can lose some privacy with respect to an optimal unbounded mechanism in exchange for a better worst-case quality loss guarantee by enforcing the bounding constraint $\Qwc(f,\pi)\leq\Qwcmax$. 

\subsection{Other Complementary Metrics}

Now, we finally outline other metrics that can be used together with the average error and average quality loss to assess the privacy of mechanisms, and leave the development of mechanisms taking them into account as subject for future work.

Geo-indistinguishability \eqref{eq:geoindcondition} inherently ensures that an input location $x$ is mapped to a nearby location with more probability than to a far location, which solves the privacy issue we illustrated with the coin mechanism. However, this privacy notion is not compatible with a worst-case quality loss constraint by definition, due to the fact that $f(z|x)>0$ implies $f(z|x')>0$, $\forall x'\in\Xalph$. A possible approach to solve this utility issue of geo-indistinguishability can be to relax its definition, allowing a small tolerance value $\Delta\ll 1$, i.e.,
\begin{equation}
 \int_{A} f(z|x) dz\leq e^{\epsilon\cdot \dP{x,x'}} \cdot \int_{A} f(z|x') dz + \Delta \,,\quad\begin{aligned}\forall x,x'\in\Xalph\,,\\ \forall A\subseteq\RR\,.\end{aligned}
\end{equation}

Other interesting metrics to assess the privacy of mechanisms are those based on the worst-case output. For example, the worst-case output average error, defined as
\begin{equation}
 \PWCAE(f,\pi)=\min_{\substack{z\in\RR\\ f_Z(z)>0}} \min_{\hat{x}\in\RR}\left\{ \sum_{x\in\Xalph} \pi(x) \cdot f(z|x) \cdot \dP{x,\hat{x}} \right\}\,,
\end{equation}
measures the average error of the adversary's estimation in the most vulnerable output. When applied to the coin mechanism, this metric would reveal its privacy issue, since $\PWCAE(\fcoin,\pi)=0$.

On the other hand, the worst-case output conditional entropy, defined as
\begin{equation}
 \PWCCE(f,\pi)=\min_{\substack{z\in\RR\\ f_Z(z)>0}} \sum_{x\in\Xalph} p(x|z) \cdot \log p(x|z)\,,
\end{equation}
reveals the uncertainty the adversary has after observing $z$ in the worst case (for the user). If there is any output value $z$ that leaks a lot of information about the real location $x$ (as it happens with every $z\neq z^*$ in the coin mechanism), this metric highlights it.

The metrics introduced throughout this section add additional dimensions to the privacy and quality loss evaluation procedure, revealing features not captured by the standard 2-dimensional approach based on the average error and the average loss. An example of this new characterization of privacy is shown in Fig.~\ref{fig:2D_3D_paradigms} where we show the performance of two mechanisms as a 3-D plot of $\PAE$, $\PCE$ and $\Qavg$, together with the projections in the $\PAE$-$\Qavg$ and $\PCE$-$\Qavg$ planes. In the next section, we show similar examples (albeit with 2-dimensional plots, for clarity) of particular location privacy preserving mechanisms. 

\begin{figure}[t]
  \centering
  \includegraphics[width=\columnwidth]{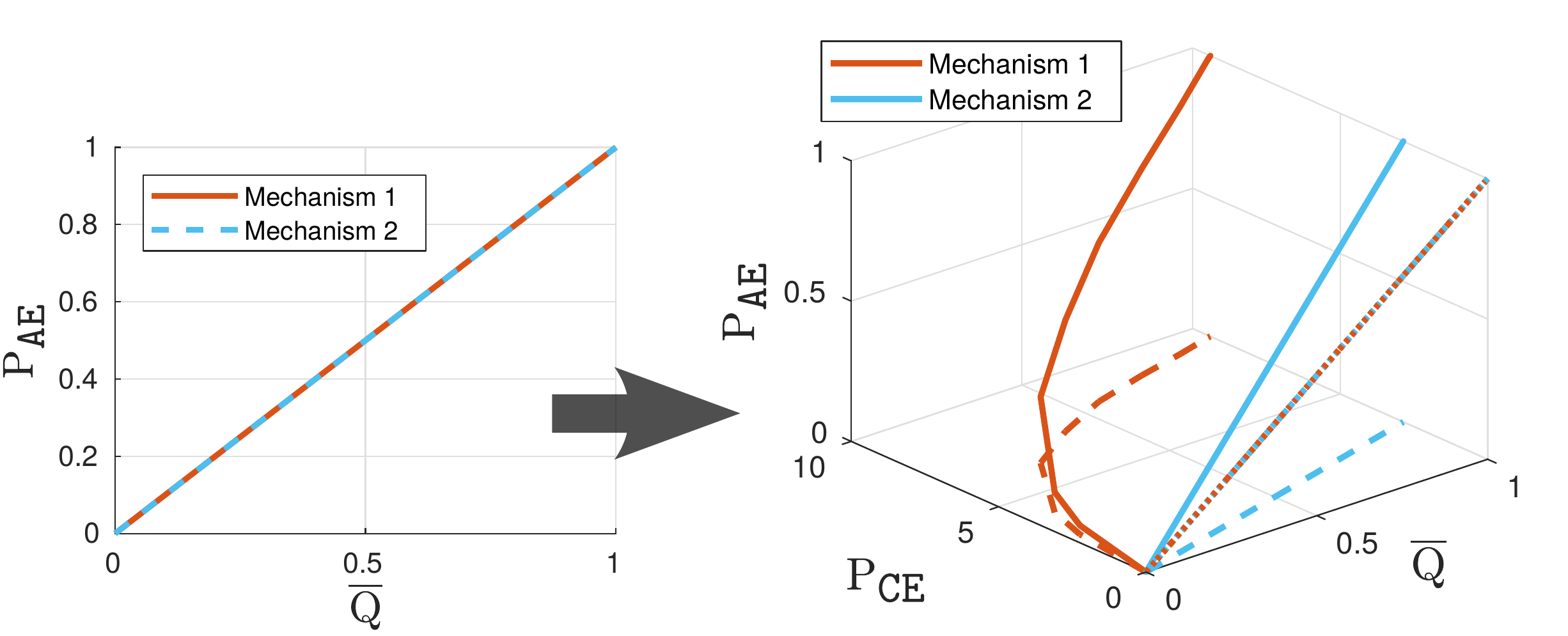}
  \caption{Two mechanisms that perform equally in the $\PAE$ vs.~$\Qavg$ plane, might behave very differently in practice. This is revealed by considering a multi-dimensional characterization of privacy.}
  \label{fig:2D_3D_paradigms}
\end{figure}

\section{Evaluation}
\label{sec:eval}

In this section, we assess the performance of different location privacy-preserving mechanisms with respect to different privacy notions. 
Our experiments confirm that relying on a single metric for evaluation can lead to an erroneous assessment of the privacy provided by a mechanism.
We divide our evaluation into two parts. First, we consider the continuous scenario introduced in Section~\ref{sec:sysmodel} and use real datasets to evaluate the performance of unbounded mechanisms, and of mechanisms that guarantee a maximum worst-case quality loss. Second, we consider a simpler scenario where the locations can only belong to a discrete set, and evaluate other defenses that have been proposed in the literature. All our experiments are performed using Matlab.\footnote{\url{https://www.mathworks.com/products/matlab.html}}

\subsection{Continuous Scenario}

\begin{figure*}[t]
  \begin{minipage}[b]{0.27\linewidth}
    \centering
    \includegraphics[width=\textwidth]{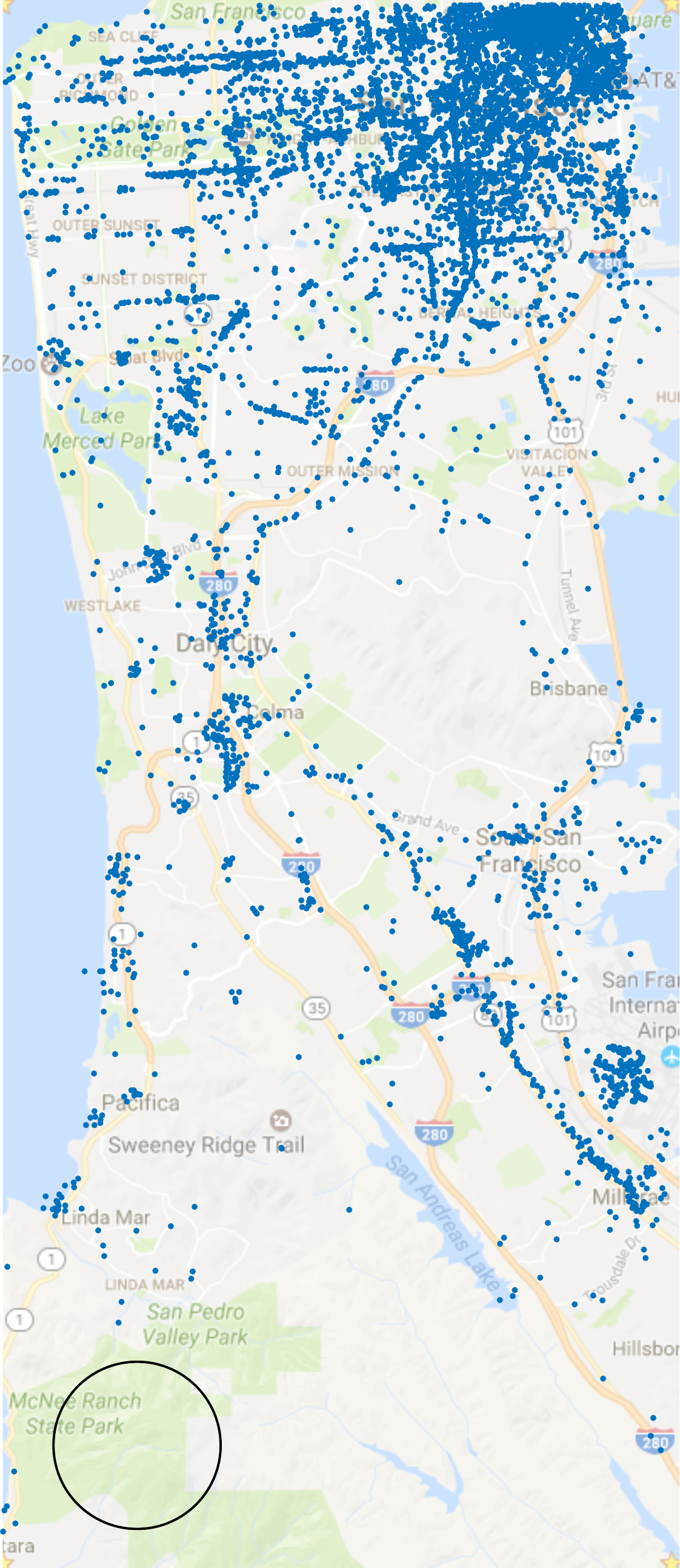}
    \caption{Points of interest in the San Francisco region taken from Gowalla dataset.}
    \label{fig:PoI_Gowalla}
  \end{minipage} \hfill
  \begin{minipage}[b]{0.68\linewidth}
    \centering
    \begin{minipage}[b]{\linewidth}
      \includegraphics[width=0.5\textwidth]{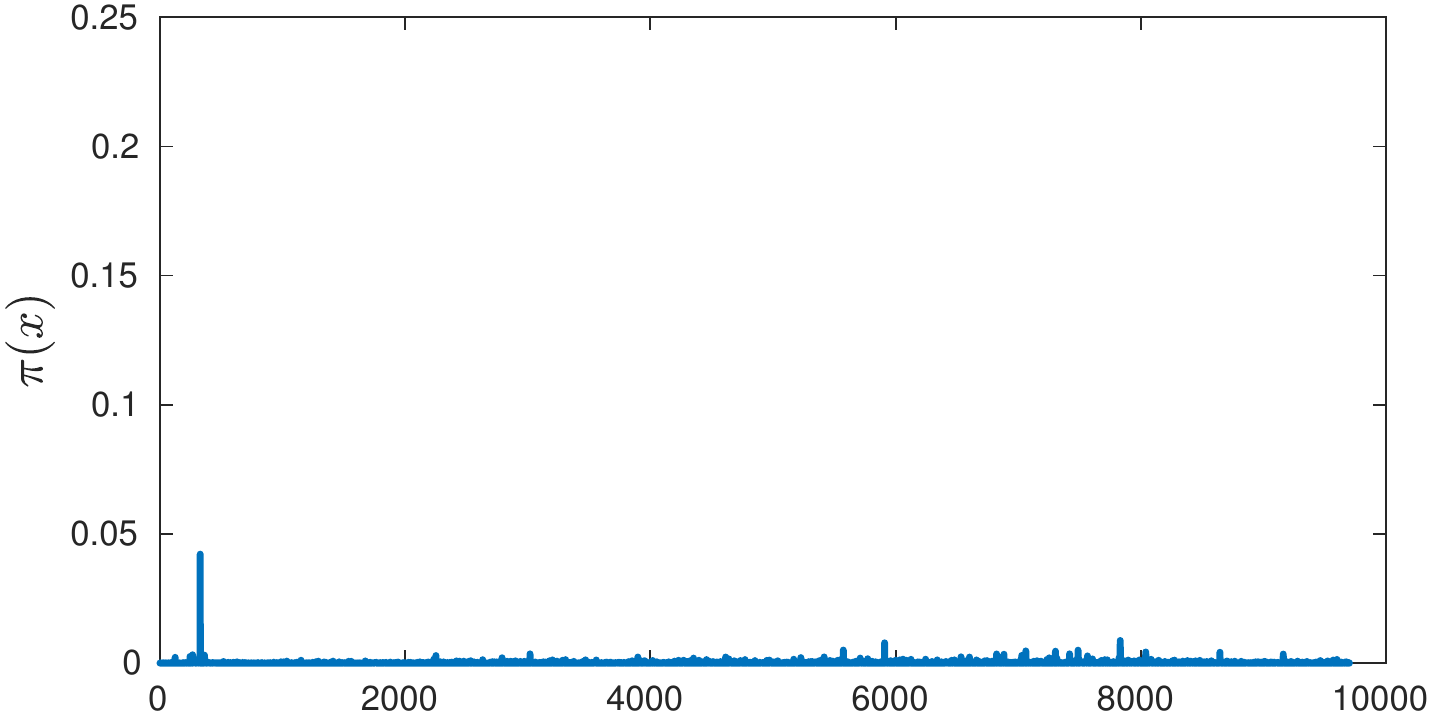}
      \includegraphics[width=0.5\textwidth]{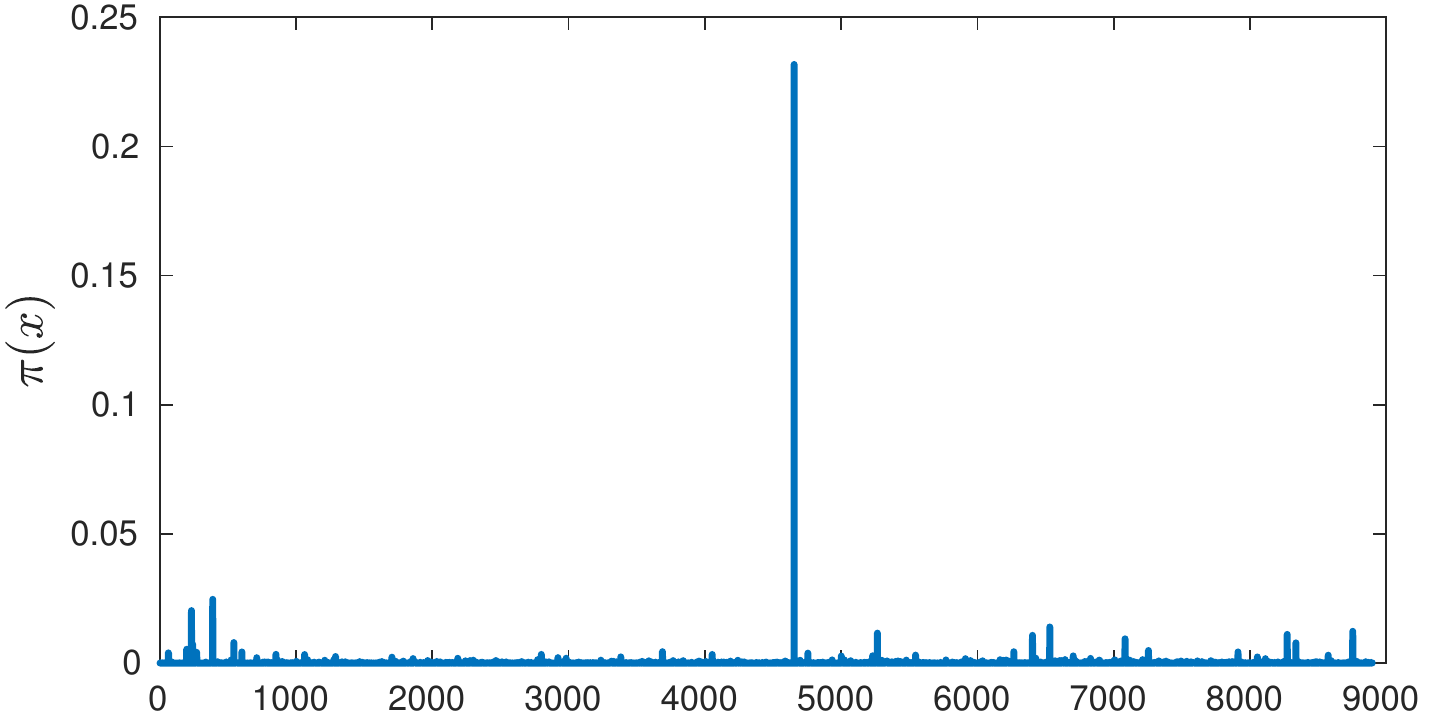}
      \caption{Priors $\pi$ for Gowalla (left) and Brightkite (right) datasets.}
      \label{fig:priors}
    \end{minipage}
    \begin{minipage}[b]{\linewidth}
      \includegraphics[width=0.48\textwidth]{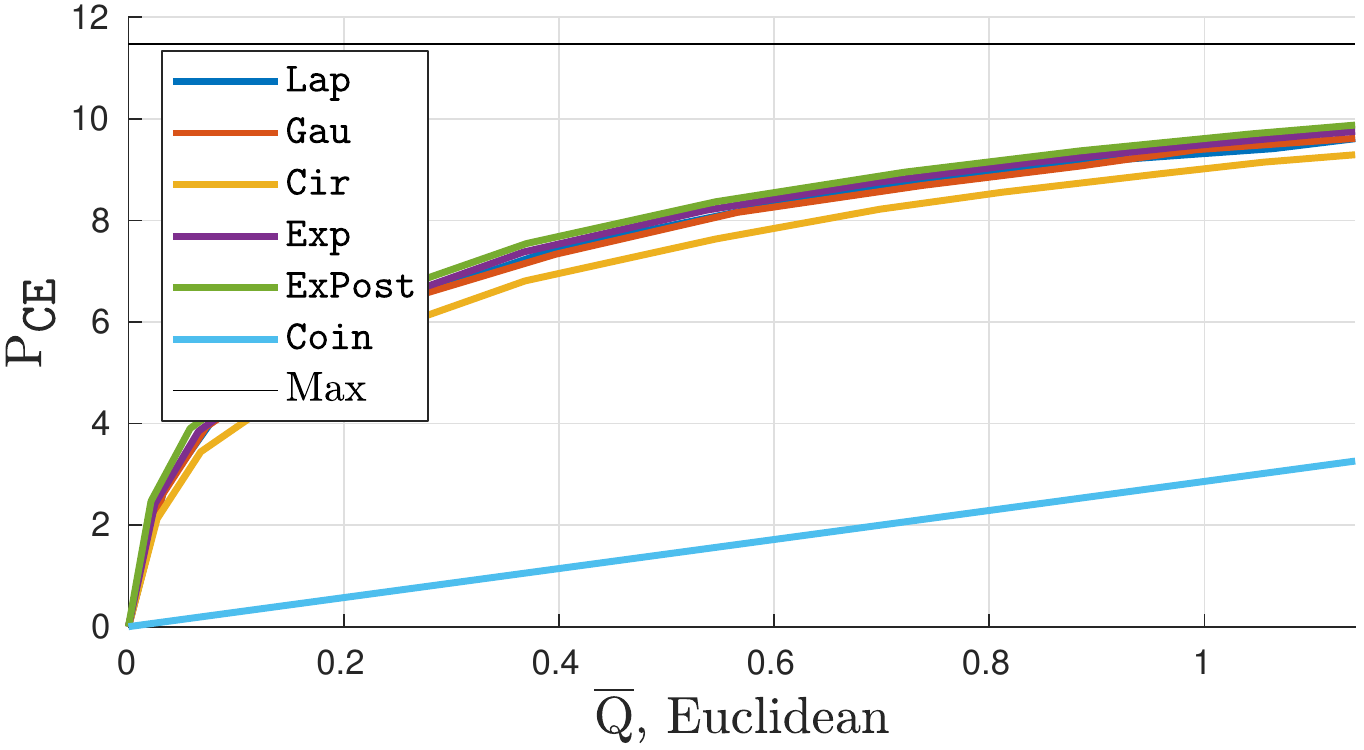} \hfill
      \includegraphics[width=0.48\textwidth]{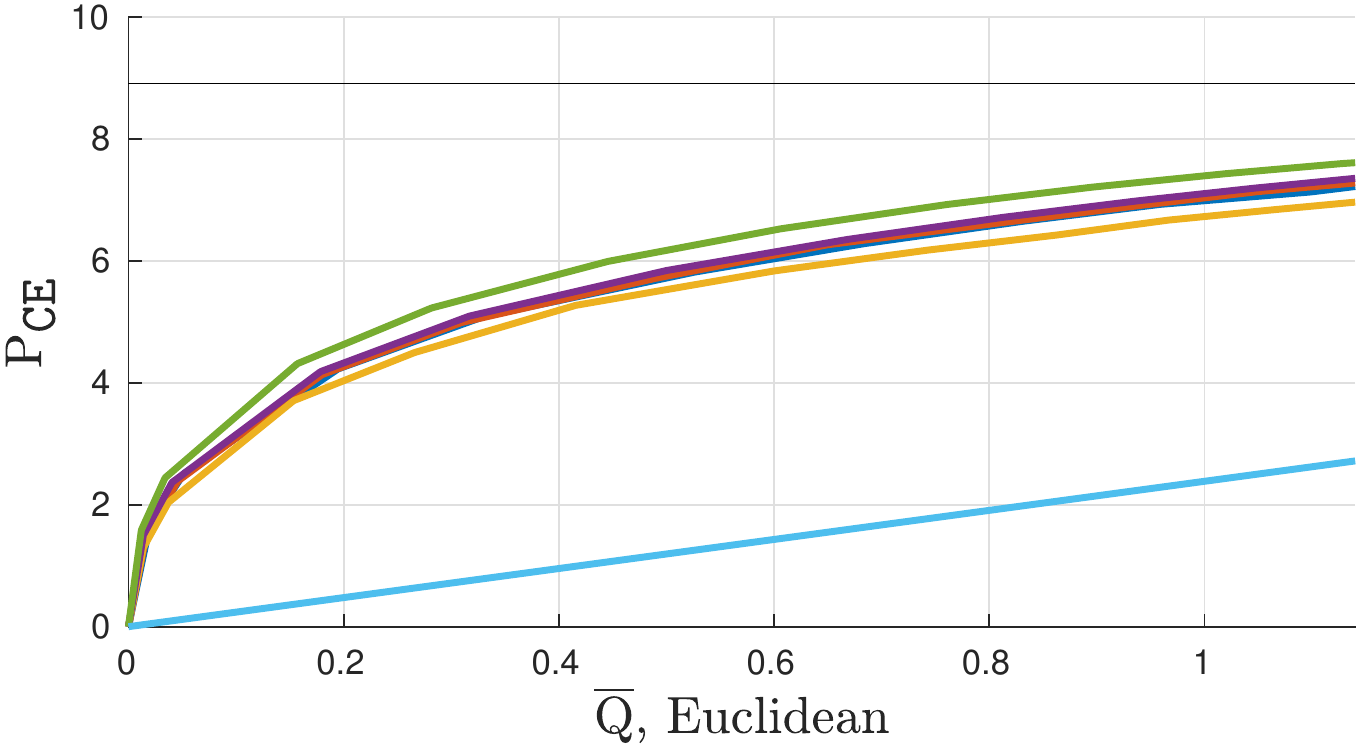}
      \caption{Conditional entropy vs.~average quality loss for Gowalla (left) and Brightkite (right) datasets.}
      \label{fig:Unbounded_CE}
    \end{minipage}
    \begin{minipage}[b]{\linewidth}
      \includegraphics[width=0.5\textwidth]{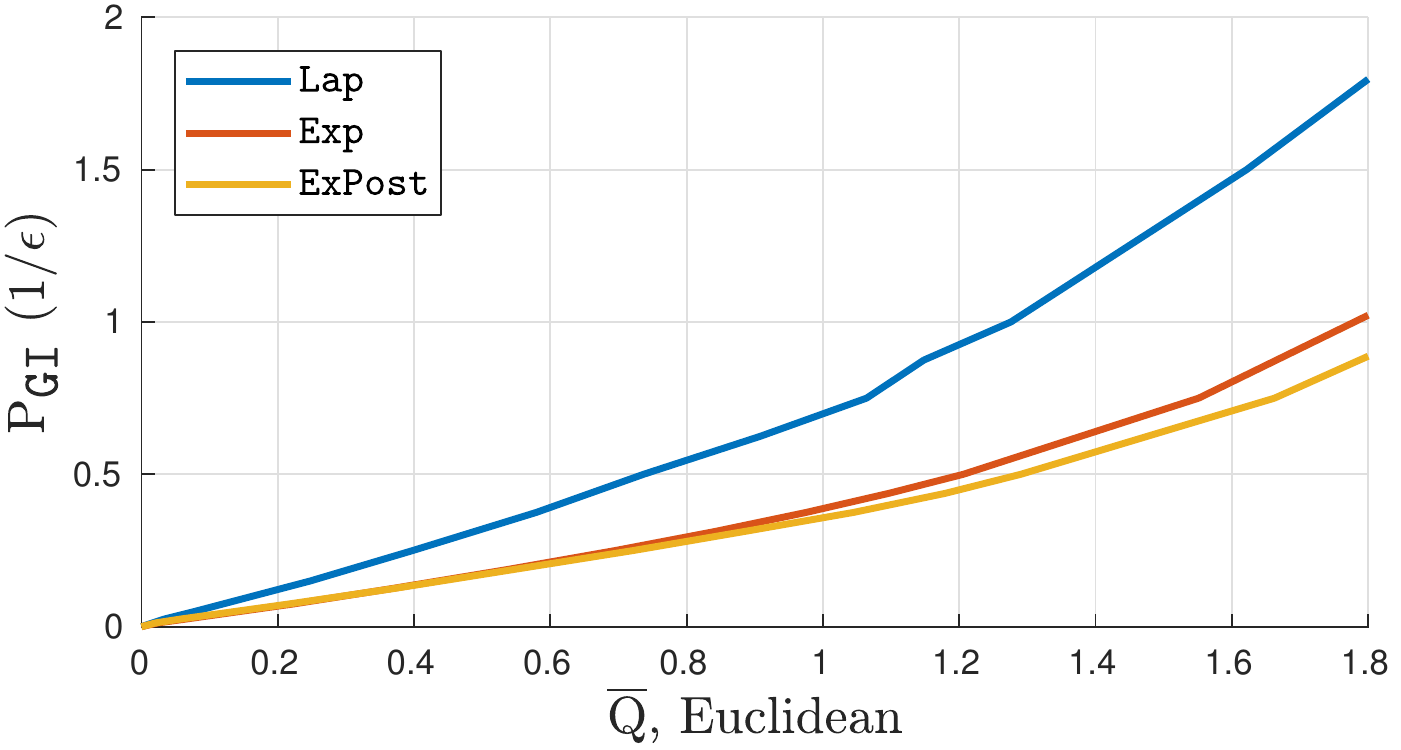}
      \includegraphics[width=0.5\textwidth]{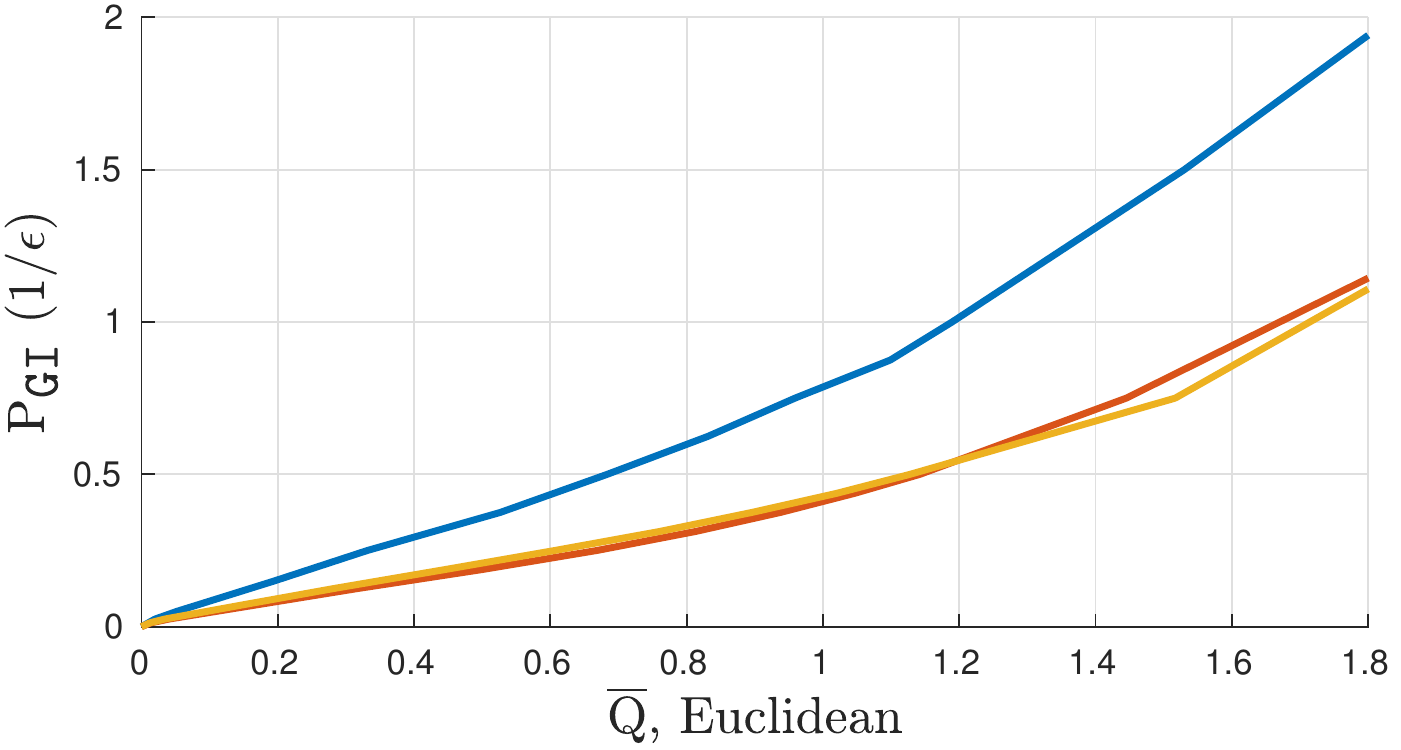}
      \caption{Geo-Ind Privacy $\PGI$ vs.~average quality loss for Gowalla (left) and Brightkite (right) datasets.}
      \label{fig:Unbounded_GI}
    \end{minipage}
  \end{minipage}
\end{figure*}

For this part of the evaluation, we consider that users are interested in querying about Points of Interest (PoIs) in a discrete set but they can report any point in $\RR$ to the server (see Section~\ref{sec:sysmodel}). We also consider that the adversary performs her estimation in $\RR$. We build the set of PoIs using the Gowalla\footnote{\url{https://snap.stanford.edu/data/loc-gowalla.html}} and Brightkite\footnote{\url{https://snap.stanford.edu/data/loc-brightkite.html}} real-world datasets. Following the approach of the finite domain evaluation in \cite{Practical2016}, we restrict the PoIs to a finite region of San Francisco area between the latitude coordinates ($37.5395$ and $37.7910$) and longitude ($-122.5153$ and $-122.3789$). We choose the San Francisco area because it contains a big density of points of interest and a large number of user check-ins, which ensures that the data is rich and representative of what one would expect from users living in the area. On the other hand, considering a finite region allows us to evaluate mechanisms whose computational cost increases with the number of points of interest, such as the exponential and exponential posterior mechanisms. We transform the PoIs into Cartesian coordinates in kilometers using the Haversine formula with respect to the center of the region. We end up with $|\Xalph|=9\,701$ PoIs for Gowalla and $|\Xalph|=8\,898$ for Brightkite, distributed in an area of roughly $28\text{km}\times 12\text{km}$. As example, the distribution of PoIs for Gowalla is shown in Fig.~\ref{fig:PoI_Gowalla}. For each dataset, we compute the prior $\pi$ by counting how many users check-in on each point of interest and normalizing the resulting histogram. The obtained priors are shown in Fig.~\ref{fig:priors}. We see that, in both datasets, there is a single point of interest $x_{\text{top}}$ that draws a lot of attention from the users ($\pi(x_{\text{top}})\approx0.04$ in Gowalla and $\pi(x_{\text{top}})\approx 0.23$ in Brightkite).

\begin{figure*}[t]
  \begin{minipage}[b]{0.37\linewidth}
    \centering
    \includegraphics[width=\textwidth]{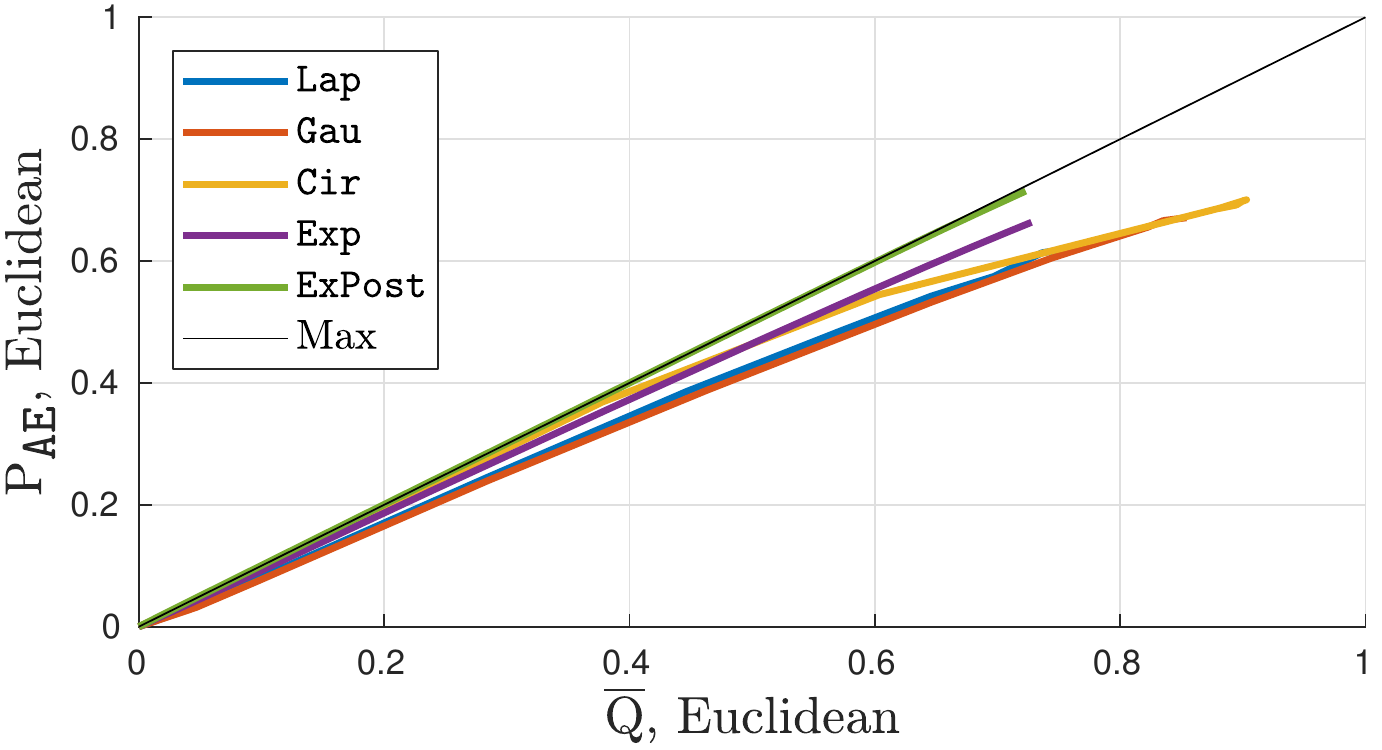}\\
    \includegraphics[width=\textwidth]{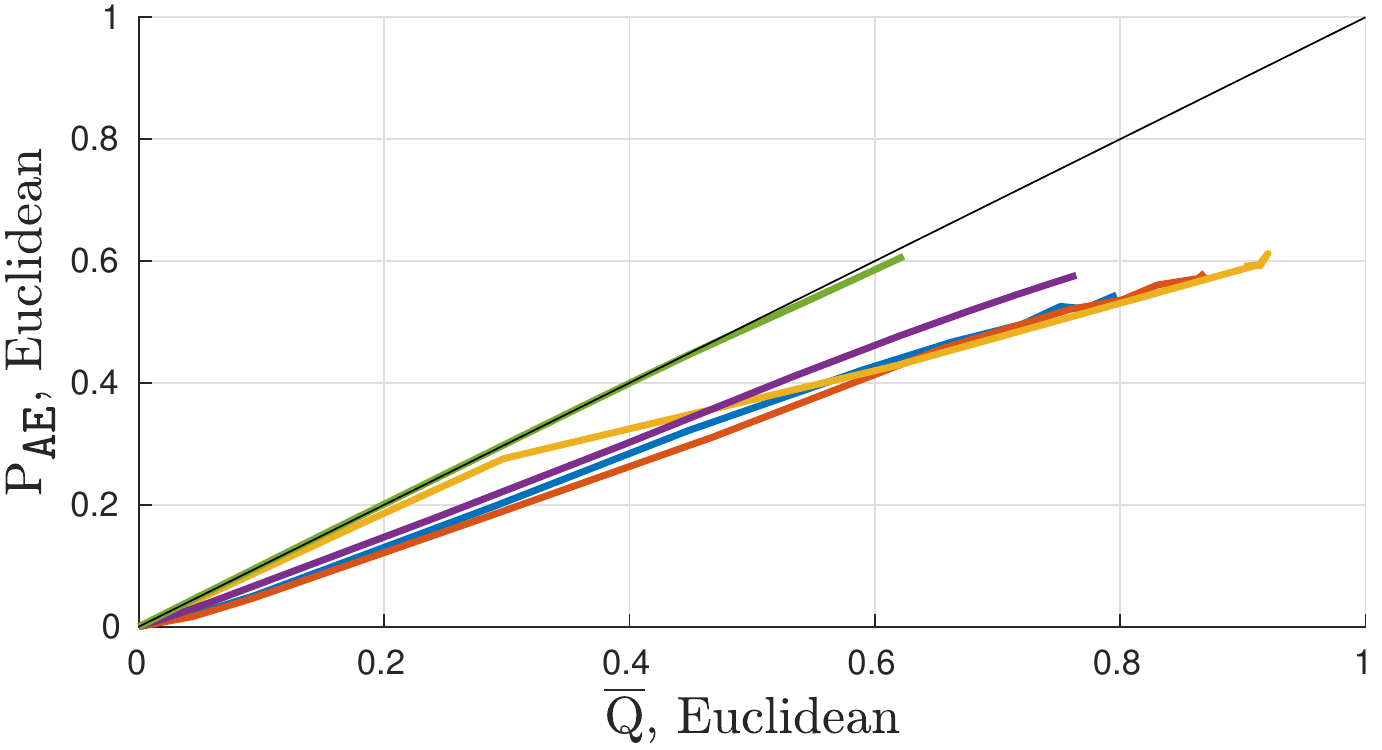}\\
    \caption{Average error vs.~average quality loss for different bounded mechanisms.}
    \label{fig:Bounded_L2}
  \end{minipage} \hfill
  \begin{minipage}[b]{0.37\linewidth}
    \centering
    \includegraphics[width=\textwidth]{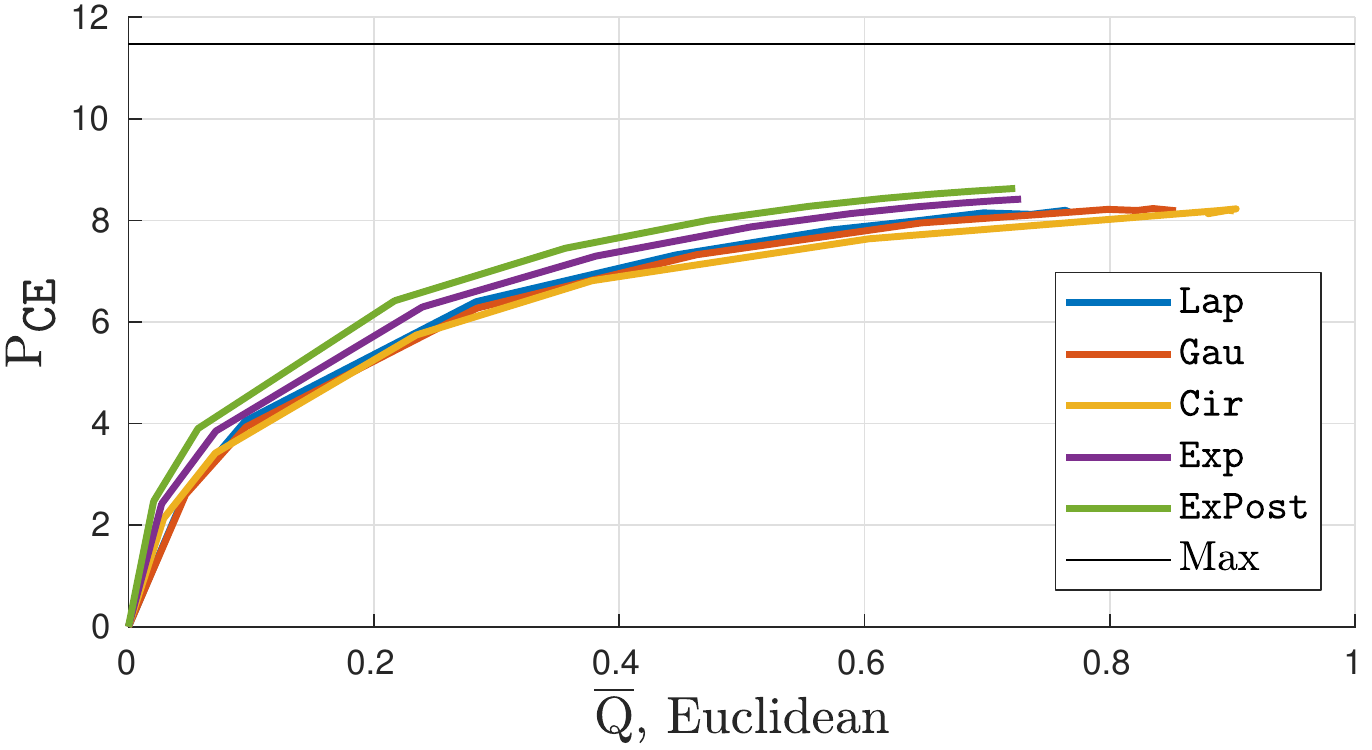}\\
    \includegraphics[width=\textwidth]{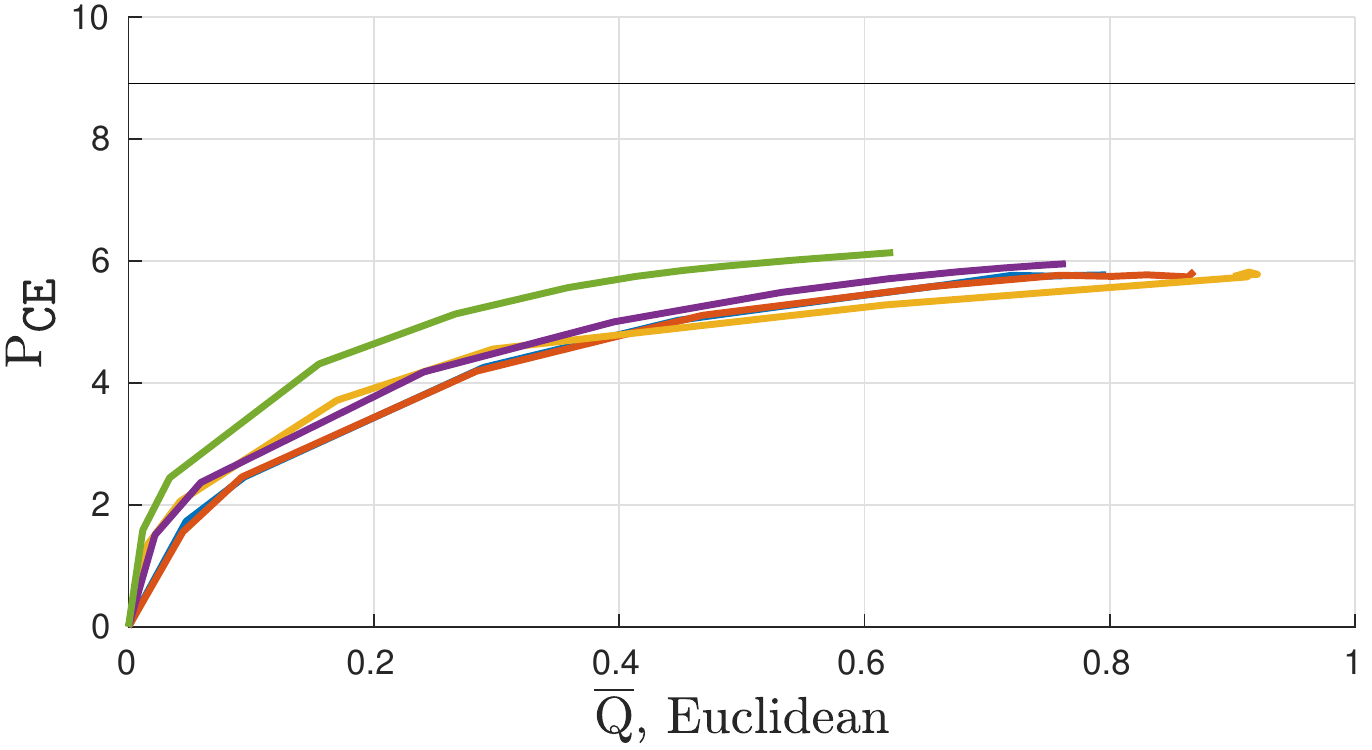}\\
    \caption{Conditional entropy vs.~average quality loss for different bounded mechanisms.}
    \label{fig:Bounded_CE}
  \end{minipage} \hfill
  \begin{minipage}[b]{0.2\linewidth}
    \centering
    \includegraphics[width=\textwidth]{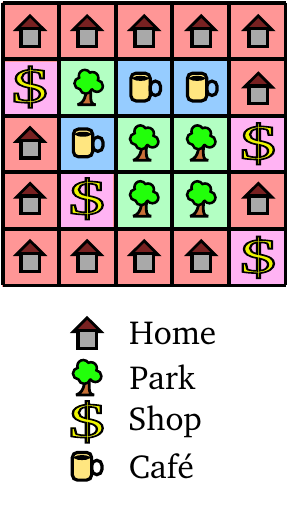}
    \caption{Semantic map of the discrete synthetic scenario.}
    \label{fig:scenario2}
  \end{minipage}
\end{figure*}

We evaluate six location-privacy preserving mechanisms, measuring their performance in terms of the average adversary error ($\PAE$), conditional entropy ($\PCE$) and geo-indistinguishability ($\PGI$) for different values of average quality loss ($\Qavg$). We always use the Euclidean distance for the quality loss $\dQ{x,z}=||x-z||_2$, and therefore the optimal remapping in \eqref{eq:optremapping} is obtained by computing the geometric median of the posterior. We compute this median using Weiszfeld's iterative method. We first evaluate the mechanisms without any bounds on their worst-case quality loss, and then imposing such constraint. 

The first three mechanisms we evaluate consist in adding noise in the continuous plane and then remapping them. We generate this noise in polar coordinates, sampling $\theta$ from a uniform distribution in ($0$, $2\pi$) and the radius $r$ from a distribution specified below. Since for these algorithms we cannot find a closed form expression for $f(z|x)$, we evaluate them empirically. To this end we sample $\pi$ to obtain $x$, we obtain $z$ adding the noise and performing the remapping, and then we measure privacy according to each metric. We report averages over $5\,000$ repetitions. These mechanisms are:
\begin{itemize}
 \item \textbf{[$\Lap$] Planar Laplacian noise} plus remapping \cite{Practical2016}. To generate the radius of the Laplace noise, we first sample $p$ uniformly in the interval $(0,1)$. Then, following~\cite{GeoInd2013CCS}, we set $r=\frac{1}{\epsilon}\left(W_{-1}\left(\frac{p-1}{e}\right)+1\right)$ where $W_{-1}$ is the $-1$ branch of the Lambert W function. We test different values of $\epsilon$ from $0.4\text{km}^{-1}$ to $40\text{km}^{-1}$, so that the average loss varies between $0.05$ and $5$km. 
 
 \item \textbf{[$\Gau$] Bi-dimensional Gaussian noise} plus remapping. To generate Gaussian noise, we sample the radius from a Rayleigh distribution, varying its mean from $0.05$ to $5$km.
 
 \item \textbf{[$\Cir$] Uniform circular noise} plus remapping. In this case, we sample the radius $r\in(0,R)$ from $f(r)=r/R^2$, where $R$ is the maximum radius of the circle, which we vary from $0.075$km to $7.5$km. This ensures an average loss that varies between $0.05$ and $5$km.
\end{itemize}

Second, we evaluate three mechanisms that output values in a discrete set, whose conditional probability density functions $f(z|x)$ can be computed arithmetically. This allows us to exactly determine their privacy and quality loss performance. These mechanisms are:
\begin{itemize}
 \item \textbf{[$\Coin$] The coin mechanism,} explained in Sect.~\ref{sec:coin}. We vary its average loss $\Qavg$ from $0$ to $2$. 
 \item \textbf{[$\Exp$] The Exponential mechanism} plus optimal remapping. The exponential mechanism is a general differential privacy technique that can be applied to provide geo-indistin\-guishability~\cite{dwork2008differential}. We set $\Zalph=\Xalph$ and set a parameter $b$, then compute the probability of mapping each input $x$ to an output $z$ as $\fdiscrete{z}{x}=a\cdot e^{-b\cdot \dQ{x,z}}$, where $a$ ensures that $\sum_{z\in\Zalph} \fdiscrete{z}{x}=1$. Then, we apply an optimal remapping to the outputs of this function and obtain $f(z|x)$. In the experiments, we vary $b$ from $0.4\text{km}^{-1}$ and $40\text{km}^{-1}$. 
 \item \textbf{[$\BA$] Exponential posterior mechanism}, proposed in Section~\ref{sec:BlahutArimoto}. In our experiments we set the discrete output alphabet of this algorithm to $\Zalph=\Xalph$.
\end{itemize}

\subsubsection{Results for unbounded mechanisms (no $\Qwc$ constraint)}
When the worst-case quality loss is not constrained, the optimal remapping ensures that all mechanisms are optimal in terms of average error, i.e., $\PAE=\Qavg$ (see Fig.~\ref{fig:Unbounded_L2} in the Appendix). This shows that the optimal remapping applied to \emph{any} mechanism achieves an optimal performance, whether it was Laplacian noise or a binary selection of a location such as $\Coin$, as we proved in Sect.~\ref{sec:PAEproblem}.

Figure~\ref{fig:Unbounded_CE} shows the mechanisms' performance in terms of conditional entropy $\PCE$, where the horizontal black line represents the maximum entropy achievable, i.e., the entropy of the prior $\pi$. Unsurprisingly, $\BA$ outperforms the rest of the mechanisms, as it is optimized with respect to this metric. The relative improvement of $\BA$ with respect to the other algorithms is slightly better in Brightkite than in Gowalla. This is due to the fact that in Brightkite the most frequent PoI is more popular than in Gowalla (see Fig.~\ref{fig:priors}), and thus performing well in this location is crucial to achieve a good overall privacy level in Brightkite. The iterative structure of $\BA$ allows this mechanism to refine its performance and be more effective than the rest of the mechanisms around this PoI. We note, however, that this refinement comes at the price of an increase in computational cost. Overall, all the mechanisms achieve a similar performance in terms of conditional entropy, except for the coin, that performs poorly. This reinforces the critique in Sect.~\ref{sec:coin}: even though $\Coin$ is optimal in terms of the average adversary error, measuring its performance in terms of conditional entropy reveals its privacy flaws.

Figure~\ref{fig:Unbounded_GI} shows the mechanisms' performance in terms of geo-indistinguishabi\-lity $\PGI(f)$ (we recall that $\PGI(f)=1/\epsilon$), only for $\Lap$, $\Exp$ and $\BA$, as these are the only algorithms that guarantee this property. As already seen in \cite{Practical2016}, the Laplace noise outperforms the exponential mechanism, and $\BA$ performs similar to the latter.

\begin{figure*}[t]
  \centering
  \subfloat[Average error (Euclidean)]{\includegraphics[width=.666\columnwidth]{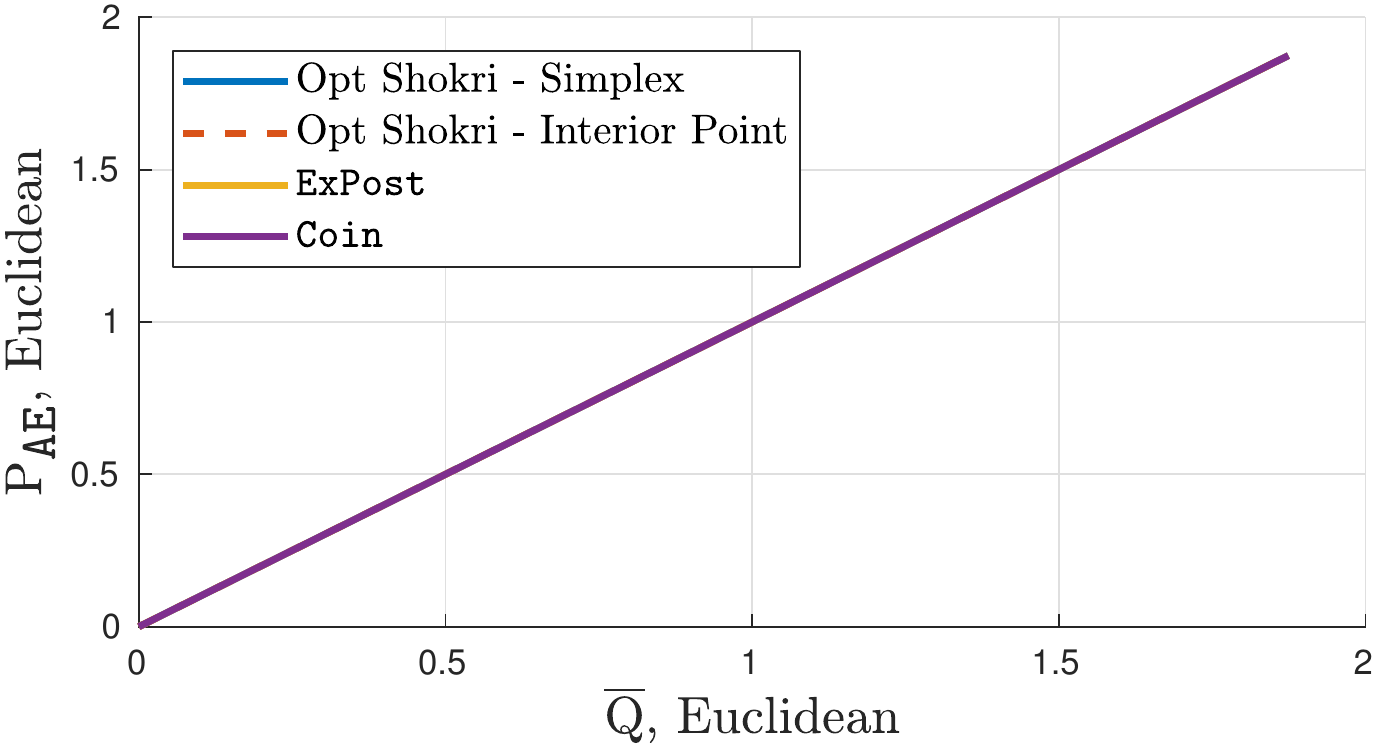}\label{fig:Simple1_L2}}
  \subfloat[Average error (semantic)]{\includegraphics[width=.666\columnwidth]{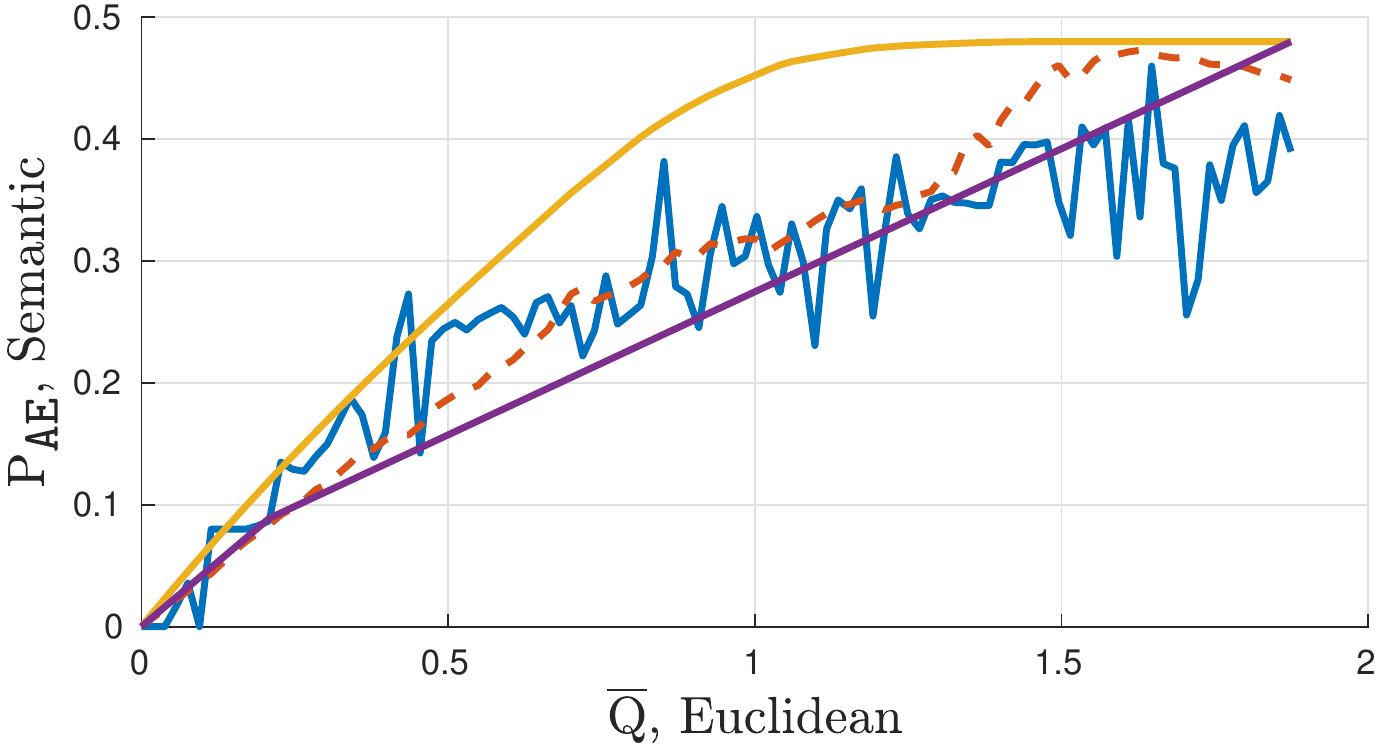}\label{fig:Simple1_S}}
  \subfloat[Conditional entropy]{\includegraphics[width=.666\columnwidth]{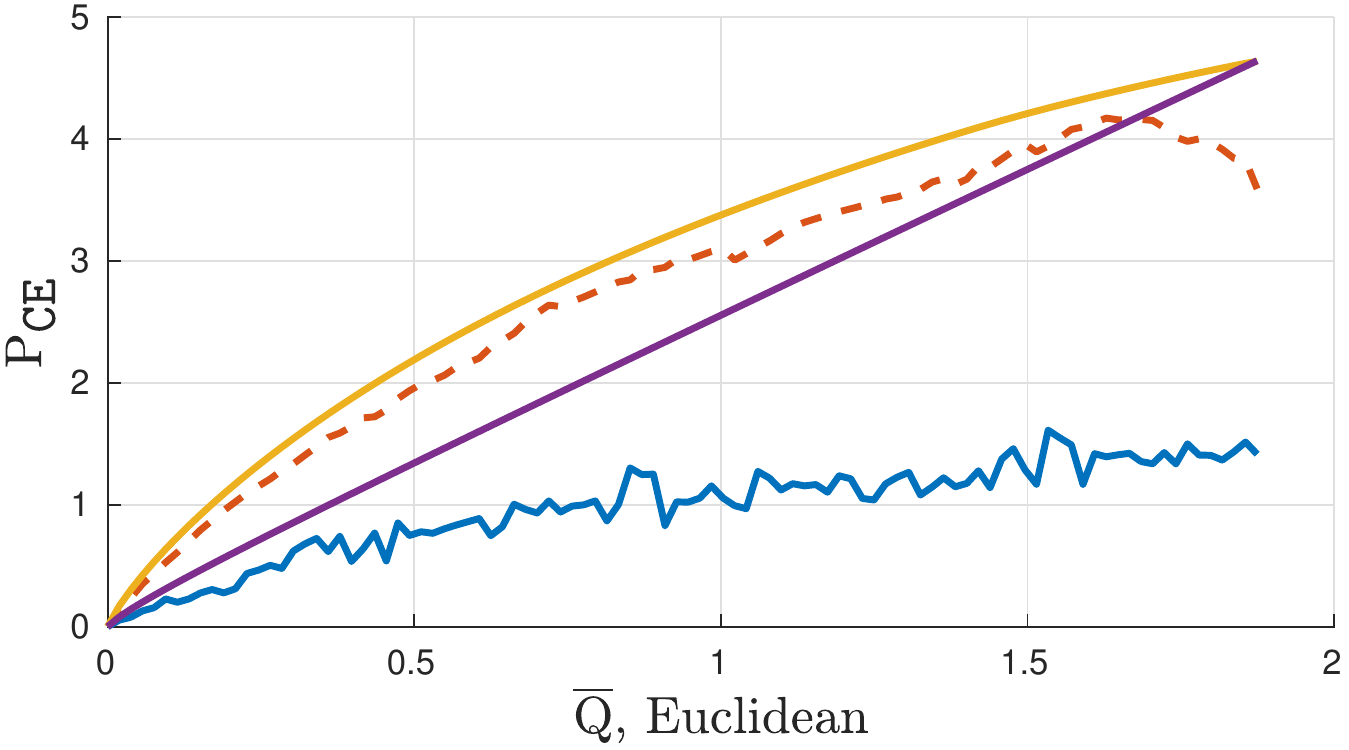}\label{fig:Simple1_CE}}
  \caption{Performance of Shokri et.~al's algorithm optimized for the adversary error in terms of Euclidean distance, compared to the coin mechanism and exponential posterior mechanism.}
  \label{fig:Simple1}
\end{figure*}

\begin{figure*}[t]
  \centering
  \subfloat[Average error (Euclidean)]{\includegraphics[width=.666\columnwidth]{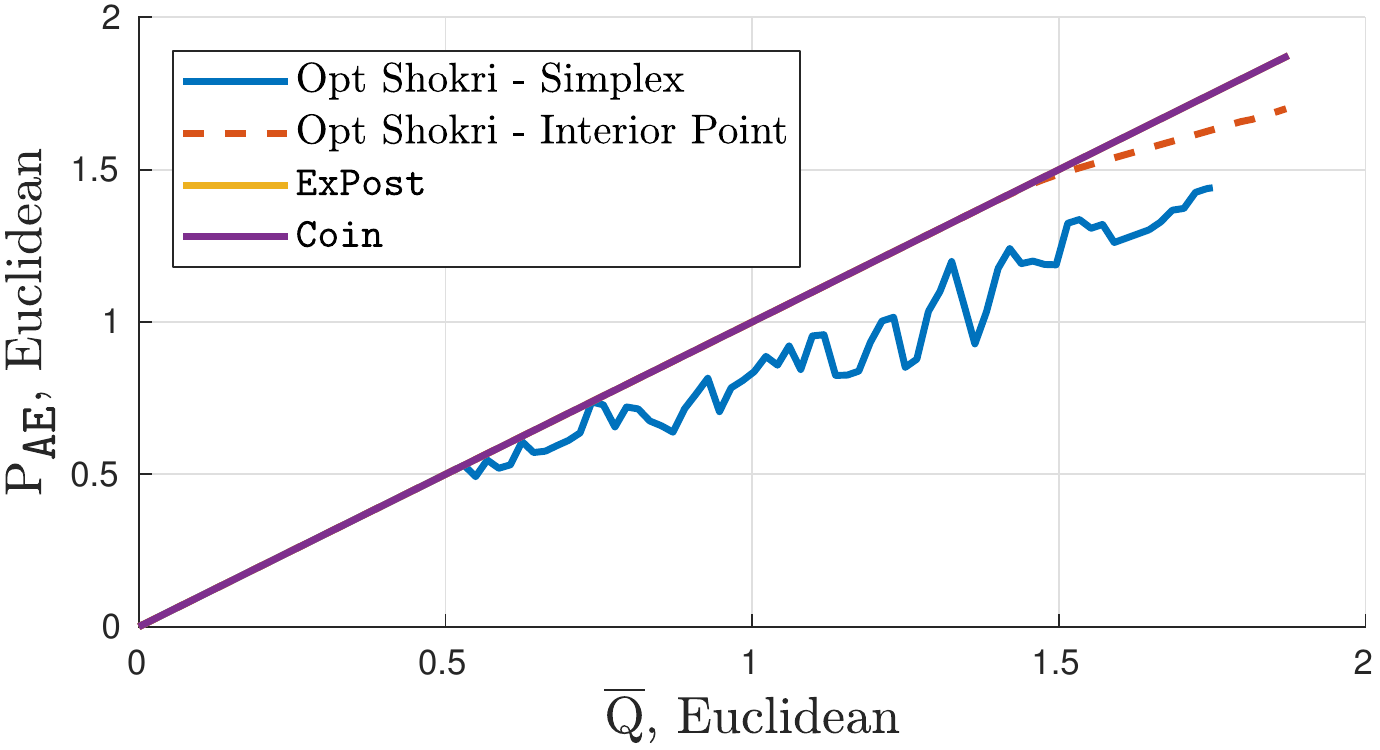}\label{fig:Simple2_L2}}
  \subfloat[Average error (semantic)]{\includegraphics[width=.666\columnwidth]{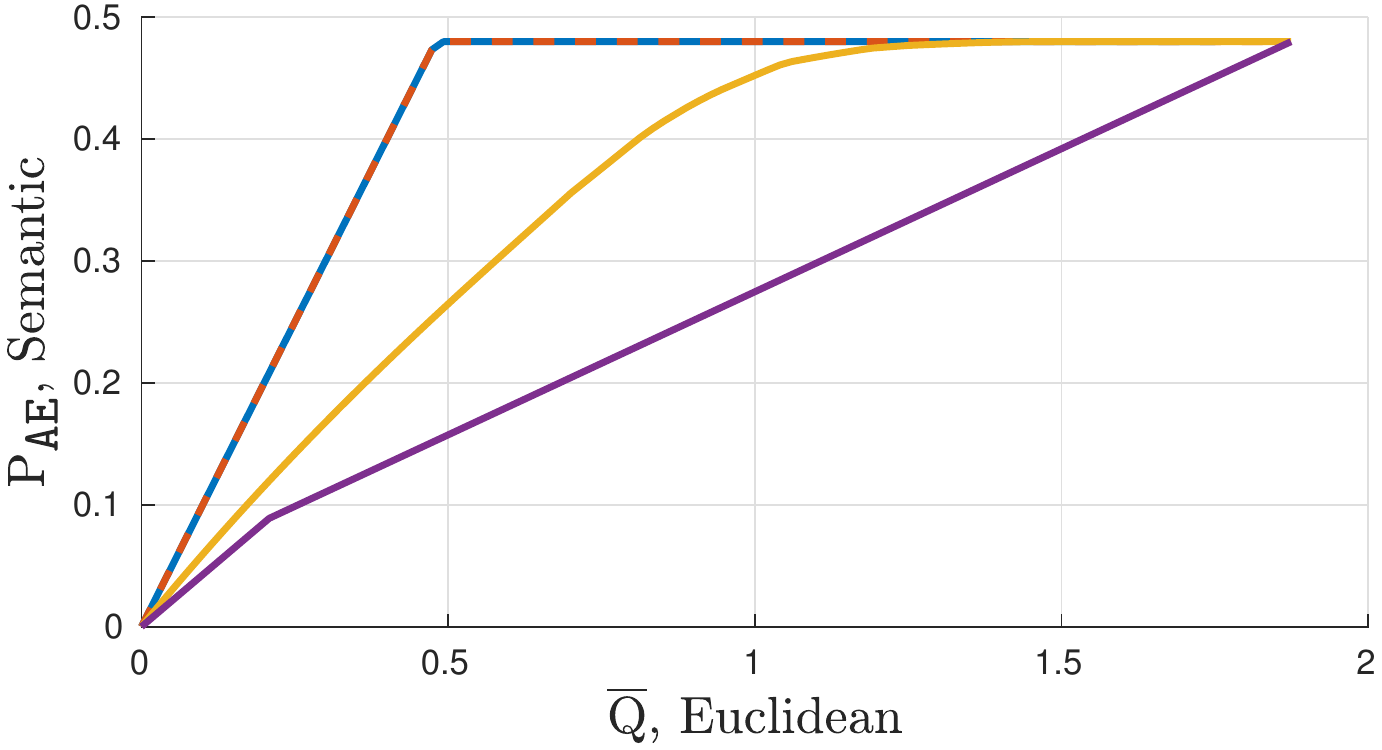}\label{fig:Simple2_S}}
  \subfloat[Conditional entropy]{\includegraphics[width=.666\columnwidth]{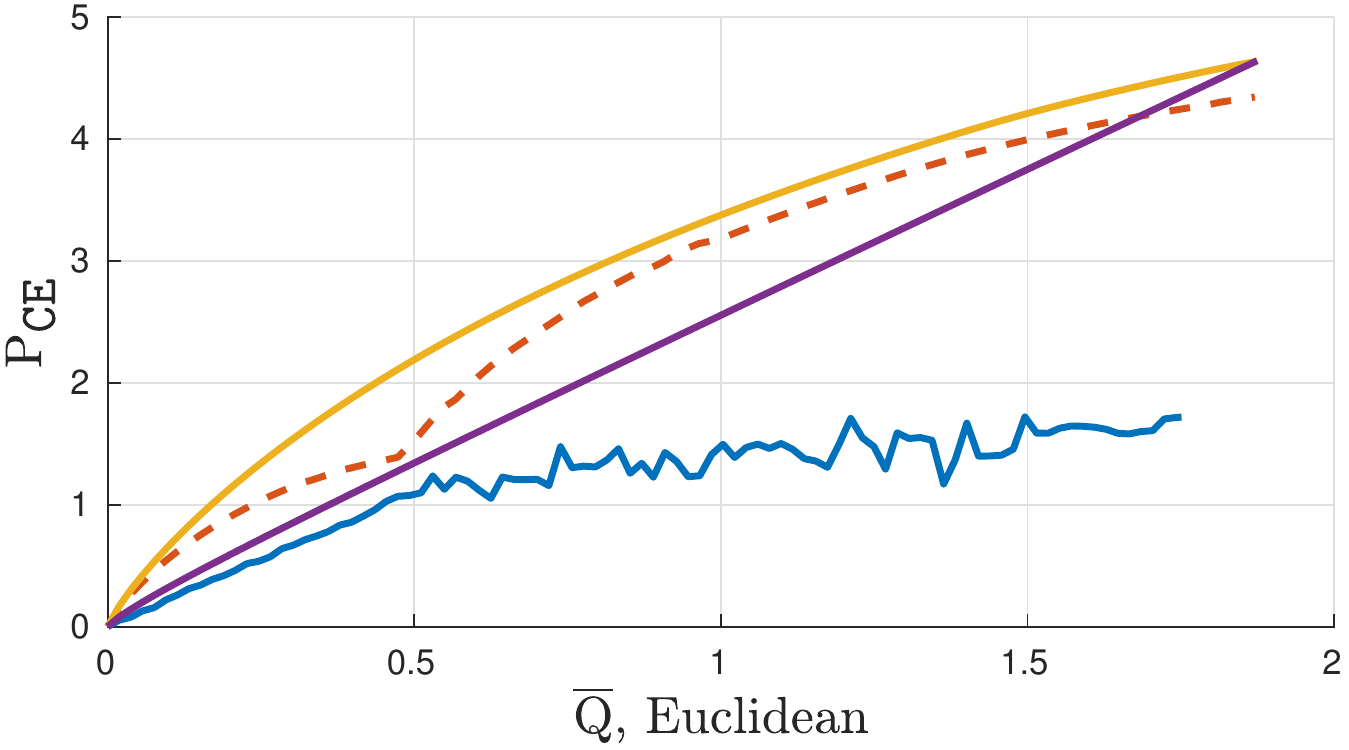}\label{fig:Simple2_CE}}
  \caption{Performance of Shokri et.~al's algorithm optimized for the adversary error in terms of semantic distance, compared to the coin mechanism and the exponential posterior mechanism.}
  \label{fig:Simple2}
\end{figure*}

\subsubsection{Results for bounded mechanisms}
We now impose a worst-case quality loss constraint of $\Qwcmax=1.5$km to the mechanisms (as a reference, we show a circle of radius $1.5$km in Fig.~\ref{fig:PoI_Gowalla}). To implement this constraint in the mechanisms, we truncate their output at $1.5$km and then apply the optimal remapping that respects the worst-case loss constraint. We do this by solving the problem in \eqref{eq:optremapping} with constraints. We do not evaluate the coin mechanism in this scenario, since it almost always violates the $\Qwc$ constraint. 

The results for the average adversary error as Euclidean distance are shown in Fig.~\ref{fig:Bounded_L2}. As expected, the mechanisms obtained after the remapping in this scenario are not necessarily optimal. We see that $\BA$ achieves a result that is close to the optimal mechanism in the unbounded case, while the other mechanisms achieve less average privacy. We conjecture this is due to the iterative nature of $\BA$, that refines its performance, while the other mechanisms are not optimized regarding the worst-case loss constraint. Again, $\BA$ achieves a wider advantage in Brightkite for the same reason explained above.

Figure~\ref{fig:Bounded_CE} shows the performance of the bounded mechanisms in terms of conditional entropy. The results are similar to those in the unbounded scenario, with $\BA$ outperforming the others with a slightly wider advantage in this case. As bounded mechanisms do not achieve geo-indistinguishability, we do not evaluate the performance with respect to this metric in this scenario. 

\subsection{Discrete scenario}

We now consider a simple synthetic scenario and evaluate the optimal mechanisms obtained following the method by Shokri et.~al \cite{Opt2012CCS}. In this work, the authors propose a linear program that finds a mechanism $f$ inside the polytope of optimal mechanisms for $\PAE$ given a constraint $\Qavg$, i.e., $f\in\Fopt{\Q}$. This approach is very versatile, as it can be computed for any pair of distance functions $\dP{\cdot}$ and $\dQ{\cdot}$. We set our synthetic scenario under the assumptions of that work: the input and output alphabets are discrete and identical $\Xalph=\Zalph$, and the adversary can only estimate locations inside that same alphabet $\Xestalph=\Xalph$. For simplicity, we consider that the set of locations in $\Xalph$ are the centers of the cells that make a $5\times 5$ square grid and assign a tag to each location that can be ``Home'', ``Park'', ``Shop'' or ``Caf{\'e}'', as depicted in Fig.~\ref{fig:scenario2}. We consider that the prior is uniform $\pi(x)=1/25\,,\,\, \forall x\in\Xalph$. We measure the point-wise loss as the Euclidean distance $\dQ{x,z}=||x-z||_2$ and consider two point-wise metrics of privacy: the Euclidean distance and a semantic distance defined as the Hamming distance between tags, i.e., $\dP{x,z}=0$ if $\text{Tag}(x)=\text{Tag}(z)$, and $\dP{x,z}=1$ otherwise. This metric is similar to the semantic metric in \cite{Semantics2016PETS}. The average error computed using this distance function represents the probability that an adversary guesses incorrectly the tag of $x$.

We evaluate $\BA$ and $\Coin$ together with the optimal mechanism proposed in \cite{Opt2012CCS}. For the latter, we solve the linear program to find optimal mechanisms in terms of maximizing $\PAE$ using the Euclidean distance (Fig.~\ref{fig:Simple1}) and the semantic distance we defined (Fig.~\ref{fig:Simple2}). As expected, the optimal mechanisms (Shokri et.~al) achieve the optimal privacy when evaluated using the adversary's error for which they are optimized (Figs.~\ref{fig:Simple1_L2} and \ref{fig:Simple2_S}), but not when evaluated against a different metric (Figs.~\ref{fig:Simple1_S} and \ref{fig:Simple2_L2}). $\BA$ and $\Coin$ achieve maximum privacy in terms of Euclidean distance, as before, but not in terms of semantic distance. This example emphasizes that optimizing a mechanism with respect to a privacy metric may provide very bad performance with respect to other privacy criteria. 

This experiment also shows another important idea: even though the solutions of the linear program both achieve approximately the same performance in terms of average error (optimal in Figs.~\ref{fig:Simple1_L2} and \ref{fig:Simple2_S}, suboptimal in Figs.~\ref{fig:Simple1_S} and \ref{fig:Simple2_L2}), they exhibit a radically different behavior in terms of conditional entropy. Indeed, using the mechanism computed with the simplex algorithm (a mechanism at a vertex of $\Fopt{\Q}$), the adversary has much less uncertainty about $x$ on average than if the user had implemented a mechanism from the interior of the polytope. This difference in entropy is also what allows us to tell apart a mechanism such as $\BA$ from $\Coin$. Note that the mechanism computed by solving the linear program with the simplex algorithm performs even worse than the coin in terms of entropy, illustrating the dangers of optimizing privacy in only one dimension.

\section{Conclusions}
\label{sec:conclusions}

In this work, we have demonstrated the problems of using a single privacy metric as indicator of the performance of location privacy preserving mechanisms. We have proven that there is more than one optimal protection mechanism in terms of maximizing the average adversary error for a given average quality loss, and that the family of mechanisms that fulfill such condition behave differently in terms of other privacy metrics. Thus, optimizing defenses with only one privacy metric in mind may lead to mechanisms that offer poor protection in other dimensions of privacy. To avoid selecting underperforming mechanisms we propose the use of complementary criteria to guide the choice. We provide two example auxiliary metrics: the conditional entropy and the worst-case loss. We propose an optimal mechanism with respect to the former, and provide means to implement mechanisms according to the latter.

We evaluate the mechanisms, comparing them to previous work, on two real datasets. Our experiments confirm two important ideas: first, that we cannot find a mechanism that performs optimally with respect to every privacy metric. Second, that even if a mechanism performs well in a particular metric it does not imply that it is necessarily beneficial for the user. Our findings reveal the need to take a step back in mechanism design to integrate privacy as a multi-dimensional notion, in order to avoid solutions that provide a false perception of privacy.

\appendix
\section{Appendix}

\subsection{Proof of Theorem~\ref{theo:optmech}}
In order to prove this result, first notice that, when $\dP{\cdot}\equiv\dQ{\cdot}$, the quality loss $\Qavg$ is an upper bound of privacy $\PAE$:
\begin{align}
  \PAE(f,\pi)&=\int_{\RR} \min_{\hat{x}\in\RR}\left\{ \sum_{x\in\Xalph} \pi(x) \cdot f(z|x) \cdot \dP{x,\hat{x}} \right\} dz \nonumber\\
      &\leq \int_{\RR} \left\{ \sum_{x\in\Xalph} \pi(x) \cdot f(z|x) \cdot \dQ{x,z} \right\} =\Qavg(f,\pi)\,,  \label{eqapp:bound}
\end{align}

Now, assume that $f'=f\circ g$, and therefore
\begin{equation}
 z = \underset{z'\in\RR}{\text{argmin }} \sum_{x\in\Xalph} \pi(x)\cdot f'(z|x)\cdot \dQ{x,z'}\,.
\end{equation}

The optimal adversary estimation of $x$ given $z$ given in \eqref{eq:advest} can be written as
\begin{equation}
 \hat{x}(z)=\underset{\hat{x}\in\RR}{\text{argmin }} \sum_{x\in\Xalph} \pi(x)\cdot f'(z|x)\cdot \dP{x,\hat{x}}\,.
\end{equation}

We see that since $\dP{\cdot}\equiv\dQ{\cdot}$ the optimal adversary estimation is doing nothing, i.e., $\hat{x}(z)=z$. This implies that $\PAE(f',\pi)=\Qavg(f',\pi)$, and since we have achieved the upper bound on privacy given in \eqref{eqapp:bound}, $f'$ is optimal.

\subsection{Geo-indistinguishability of the posterior exponential mechanism.}
We recall that the geo-indistinguishability guarantee requires the following condition to be fulfilled (now written for discrete mechanisms, where $\fdiscrete{z}{x}$ denotes the probability of reporting $z$ when the original location is $x$):
\begin{equation} \label{eqapp:geoindcondition}
 \fdiscrete{z}{x} \leq e^{\epsilon\cdot \dP{x,x'}} \cdot \fdiscrete{z}{x'}\,,\quad\forall x,x'\in\Xalph,\,z\in\Zalph\,,
\end{equation}
where $\dP{x,x'}$ is the Euclidean distance.

The last iteration of the $\BA$ algorithm in \ref{sec:BlahutArimoto} returns a mechanism that can be written for a particular input $x$ and output $z$ as
\begin{equation}
 \fdiscrete{z}{x}=\begin{cases}
         \frac{P_Z(z)\cdot e^{-b \cdot \dQ{x,z}}}{\sum_{z'\in\Zalph} P_Z(z')\cdot e^{-b \cdot \dQ{x,z'}} }\,&\text{if }P_Z(z)>0\,,\\
         0\,,&\text{if }P_Z(z)=0\,.
        \end{cases}
\end{equation}
where $\dQ{x,z}$ is the Euclidean distance. In the second case, the geo-indistinguishability guarantee is trivially achieved since given any pair of input locations $x,x'\in\Xalph$, $\fdiscrete{z}{x}=\fdiscrete{z}{x'}=0$. For the first case, we use the triangular inequality $\dQ{x,z}+\dQ{x',z}\geq\dQ{x,x'}$ to write
\begin{align}
 \fdiscrete{z}{x}=&\frac{P_Z(z)\cdot e^{-b \cdot \dQ{x,z}}}{\sum_{z'\in\Zalph} P_Z(z')\cdot e^{-b \cdot \dQ{x,z'}} }\\
       \leq&\frac{P_Z(z)\cdot e^{b \cdot \dQ{x,x'}} \cdot e^{-b \cdot \dQ{x',z}}}{\sum_{z'\in\Zalph} P_Z(z')\cdot e^{-b \cdot \dQ{x,z'}} }\\
       \leq&\frac{P_Z(z)\cdot e^{b \cdot \dQ{x,x'}} \cdot e^{-b \cdot \dQ{x',z}}}{\sum_{z'\in\Zalph} P_Z(z')\cdot e^{-b \cdot \dQ{x,x'}} \cdot e^{-b \cdot \dQ{x',z'}} }\\
       =&\frac{P_Z(z) \cdot e^{-b \cdot \dQ{x',z}}}{\sum_{z'\in\Zalph} P_Z(z')\cdot e^{-b \cdot \dQ{x',z'}} } \cdot e^{2b \cdot \dQ{x,x'}}\\
       =&e^{2b \cdot \dQ{x,x'}} \cdot \fdiscrete{z}{x'} \,,
\end{align}
which satisfies the geo-indistinguishability for $\epsilon=2b$ or $\PGI=1/2b$, if $\dQ{\cdot}$ is the Euclidean distance. This concludes the proof.

\subsection{Performance of the unbounded mechanisms in terms of the average error}
When the average error (Euclidean) and the average quality loss (Euclidean) are used to evaluate the performance of the mechanisms described in Section~\ref{sec:eval}, we achieve the trivial result $\PAE=\Qavg$. This is shown in Fig.~\ref{fig:Unbounded_L2} for completeness.

\begin{figure}[t]
  \centering
  \subfloat[Gowalla]{\includegraphics[width=.8\columnwidth]{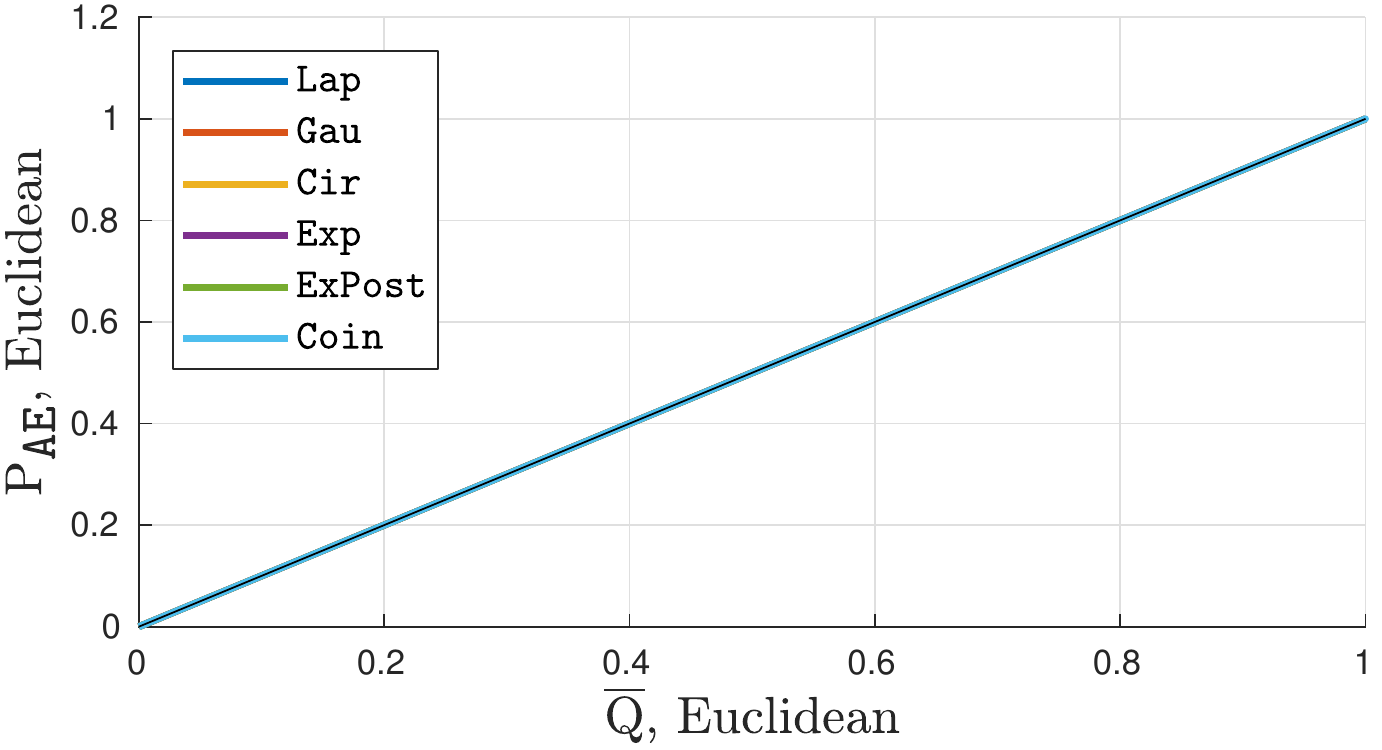}} \\
  \subfloat[Brightkite]{\includegraphics[width=.8\columnwidth]{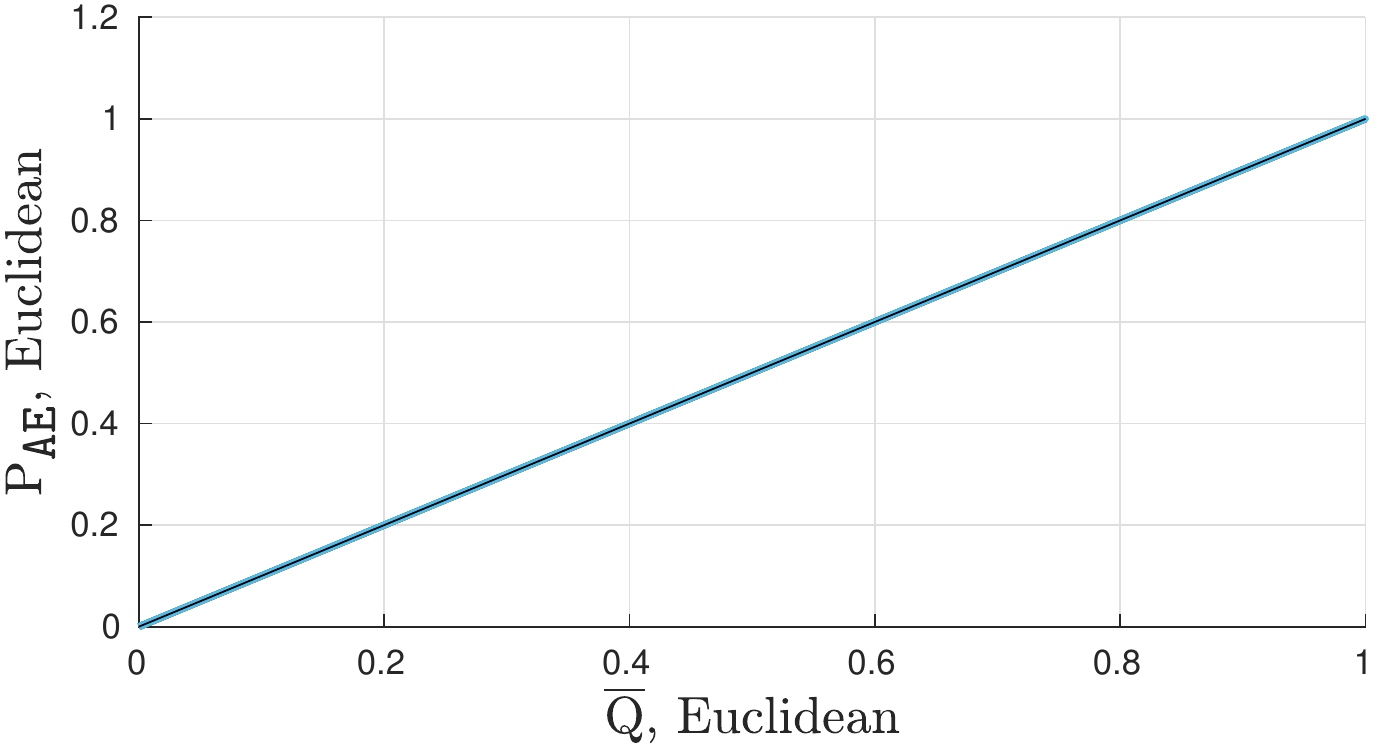}}
  \caption{Average error vs.~average quality loss for different unbounded mechanisms.}
  \label{fig:Unbounded_L2}
\end{figure}

\begin{acks}
This work is partially supported by EU H2020-ICT-10-2015 NEXTLEAP (GA n 688722), the Agencia Estatal de Investigaci{\'o}n (Spain) and the European Regional Development Fund (ERDF) under projects WINTER (TEC2016-76409-C2-2-R) and COMONSENS (TEC2015-69648-REDC), and by the Xunta de Galicia and the European Union (European Regional Development Fund - ERDF) under projects Agrupaci{\'o}n Estrat{\'e}xica Consolidada de Galicia accreditation 2016-2019 and Red Tem{\'a}tica RedTEIC 2017-2018. 
Simon Oya is funded by the Spanish Ministry of Education, Culture and Sport under the FPU grant.
\end{acks}

\bibliographystyle{ACM-Reference-Format}
\bibliography{bibliolocation}


\begin{thebibliography}{00}


\ifx \showCODEN    \undefined \def \showCODEN     #1{\unskip}     \fi
\ifx \showDOI      \undefined \def \showDOI       #1{#1}\fi
\ifx \showISBNx    \undefined \def \showISBNx     #1{\unskip}     \fi
\ifx \showISBNxiii \undefined \def \showISBNxiii  #1{\unskip}     \fi
\ifx \showISSN     \undefined \def \showISSN      #1{\unskip}     \fi
\ifx \showLCCN     \undefined \def \showLCCN      #1{\unskip}     \fi
\ifx \shownote     \undefined \def \shownote      #1{#1}          \fi
\ifx \showarticletitle \undefined \def \showarticletitle #1{#1}   \fi
\ifx \showURL      \undefined \def \showURL       {\relax}        \fi
\providecommand\bibfield[2]{#2}
\providecommand\bibinfo[2]{#2}
\providecommand\natexlab[1]{#1}
\providecommand\showeprint[2][]{arXiv:#2}

\bibitem[\protect\citeauthoryear{A{\u{g}}{\i}r, Huguenin, Hengartner, and
  Hubaux}{A{\u{g}}{\i}r et~al\mbox{.}}{2016}]%
        {Semantics2016PETS}
\bibfield{author}{\bibinfo{person}{Berker A{\u{g}}{\i}r},
  \bibinfo{person}{K{\'e}vin Huguenin}, \bibinfo{person}{Urs Hengartner}, {and}
  \bibinfo{person}{Jean-Pierre Hubaux}.} \bibinfo{year}{2016}\natexlab{}.
\newblock \showarticletitle{On the Privacy Implications of Location Semantics}.
\newblock \bibinfo{journal}{{\em Proceedings on Privacy Enhancing
  Technologies\/}} \bibinfo{volume}{2016}, \bibinfo{number}{4}
  (\bibinfo{year}{2016}), \bibinfo{pages}{165--183}.
\newblock


\bibitem[\protect\citeauthoryear{Andr{\'e}s, Bordenabe, Chatzikokolakis, and
  Palamidessi}{Andr{\'e}s et~al\mbox{.}}{2013}]%
        {GeoInd2013CCS}
\bibfield{author}{\bibinfo{person}{Miguel~E Andr{\'e}s},
  \bibinfo{person}{Nicol{\'a}s~E Bordenabe}, \bibinfo{person}{Konstantinos
  Chatzikokolakis}, {and} \bibinfo{person}{Catuscia Palamidessi}.}
  \bibinfo{year}{2013}\natexlab{}.
\newblock \showarticletitle{Geo-indistinguishability: Differential privacy for
  location-based systems}. In \bibinfo{booktitle}{{\em Proceedings of the 2013
  ACM SIGSAC conference on Computer \& communications security}}. ACM,
  \bibinfo{pages}{901--914}.
\newblock


\bibitem[\protect\citeauthoryear{Beresford and Stajano}{Beresford and
  Stajano}{2003}]%
        {BeresfordS04}
\bibfield{author}{\bibinfo{person}{Alastair~R. Beresford} {and}
  \bibinfo{person}{Frank Stajano}.} \bibinfo{year}{2003}\natexlab{}.
\newblock \showarticletitle{Location Privacy in Pervasive Computing}.
\newblock \bibinfo{journal}{{\em {IEEE} Pervasive Computing\/}}
  \bibinfo{volume}{2}, \bibinfo{number}{1} (\bibinfo{year}{2003}),
  \bibinfo{pages}{46--55}.
\newblock


\bibitem[\protect\citeauthoryear{Bilogrevic, Huguenin, Mihaila, Shokri, and
  Hubaux}{Bilogrevic et~al\mbox{.}}{2015}]%
        {Bilogrevic2015predicting}
\bibfield{author}{\bibinfo{person}{Igor Bilogrevic}, \bibinfo{person}{K{\'e}vin
  Huguenin}, \bibinfo{person}{Stefan Mihaila}, \bibinfo{person}{Reza Shokri},
  {and} \bibinfo{person}{Jean-Pierre Hubaux}.} \bibinfo{year}{2015}\natexlab{}.
\newblock \showarticletitle{Predicting users' motivations behind location
  check-ins and utility implications of privacy protection mechanisms}. In
  \bibinfo{booktitle}{{\em 22nd Network and Distributed System Security
  Symposium (NDSS)}}.
\newblock


\bibitem[\protect\citeauthoryear{Bordenabe, Chatzikokolakis, and
  Palamidessi}{Bordenabe et~al\mbox{.}}{2014}]%
        {OptGeoInd2014CCS}
\bibfield{author}{\bibinfo{person}{Nicol{\'a}s~E Bordenabe},
  \bibinfo{person}{Konstantinos Chatzikokolakis}, {and}
  \bibinfo{person}{Catuscia Palamidessi}.} \bibinfo{year}{2014}\natexlab{}.
\newblock \showarticletitle{Optimal geo-indistinguishable mechanisms for
  location privacy}. In \bibinfo{booktitle}{{\em Proceedings of the 2014 ACM
  SIGSAC Conference on Computer and Communications Security}}. ACM,
  \bibinfo{pages}{251--262}.
\newblock


\bibitem[\protect\citeauthoryear{Chatzikokolakis, Elsalamouny, and
  Palamidessi}{Chatzikokolakis et~al\mbox{.}}{2016}]%
        {Practical2016}
\bibfield{author}{\bibinfo{person}{Konstantinos Chatzikokolakis},
  \bibinfo{person}{Ehab Elsalamouny}, {and} \bibinfo{person}{Catuscia
  Palamidessi}.} \bibinfo{year}{2016}\natexlab{}.
\newblock \showarticletitle{Practical Mechanisms for Location Privacy}.
\newblock  (\bibinfo{year}{2016}).
\newblock


\bibitem[\protect\citeauthoryear{Chatzikokolakis, Palamidessi, and
  Stronati}{Chatzikokolakis et~al\mbox{.}}{2015}]%
        {Elastic2015PETS}
\bibfield{author}{\bibinfo{person}{Konstantinos Chatzikokolakis},
  \bibinfo{person}{Catuscia Palamidessi}, {and} \bibinfo{person}{Marco
  Stronati}.} \bibinfo{year}{2015}\natexlab{}.
\newblock \showarticletitle{Constructing elastic distinguishability metrics for
  location privacy}.
\newblock \bibinfo{journal}{{\em Proceedings on Privacy Enhancing
  Technologies\/}} \bibinfo{volume}{2015}, \bibinfo{number}{2}
  (\bibinfo{year}{2015}), \bibinfo{pages}{156--170}.
\newblock


\bibitem[\protect\citeauthoryear{Cover and Thomas}{Cover and Thomas}{2012}]%
        {cover2012elements}
\bibfield{author}{\bibinfo{person}{Thomas~M Cover} {and} \bibinfo{person}{Joy~A
  Thomas}.} \bibinfo{year}{2012}\natexlab{}.
\newblock \bibinfo{booktitle}{{\em Elements of information theory}}.
\newblock \bibinfo{publisher}{John Wiley \& Sons}.
\newblock


\bibitem[\protect\citeauthoryear{Dwork}{Dwork}{2006}]%
        {Dwork06}
\bibfield{author}{\bibinfo{person}{Cynthia Dwork}.}
  \bibinfo{year}{2006}\natexlab{}.
\newblock \showarticletitle{Differential Privacy}. In \bibinfo{booktitle}{{\em
  Automata, Languages and Programming, 33rd International Colloquium, {ICALP}
  2006}} {\em (\bibinfo{series}{Lecture Notes in Computer Science})},
  \bibfield{editor}{\bibinfo{person}{Michele Bugliesi}, \bibinfo{person}{Bart
  Preneel}, \bibinfo{person}{Vladimiro Sassone}, {and} \bibinfo{person}{Ingo
  Wegener}} (Eds.), Vol.~\bibinfo{volume}{4052}. \bibinfo{publisher}{Springer},
  \bibinfo{pages}{1--12}.
\newblock


\bibitem[\protect\citeauthoryear{Dwork}{Dwork}{2008}]%
        {dwork2008differential}
\bibfield{author}{\bibinfo{person}{Cynthia Dwork}.}
  \bibinfo{year}{2008}\natexlab{}.
\newblock \showarticletitle{Differential privacy: A survey of results}. In
  \bibinfo{booktitle}{{\em International Conference on Theory and Applications
  of Models of Computation}}. Springer, \bibinfo{pages}{1--19}.
\newblock


\bibitem[\protect\citeauthoryear{Fawaz, Feng, and Shin}{Fawaz
  et~al\mbox{.}}{2015}]%
        {FawazFS15}
\bibfield{author}{\bibinfo{person}{Kassem Fawaz}, \bibinfo{person}{Huan Feng},
  {and} \bibinfo{person}{Kang~G. Shin}.} \bibinfo{year}{2015}\natexlab{}.
\newblock \showarticletitle{Anatomization and Protection of Mobile Apps'
  Location Privacy Threats}. In \bibinfo{booktitle}{{\em 24th {USENIX} Security
  Symposium}}, \bibfield{editor}{\bibinfo{person}{Jaeyeon Jung} {and}
  \bibinfo{person}{Thorsten Holz}} (Eds.). \bibinfo{publisher}{{USENIX}
  Association}, \bibinfo{pages}{753--768}.
\newblock


\bibitem[\protect\citeauthoryear{Fawaz and Shin}{Fawaz and Shin}{2014}]%
        {FawazS14}
\bibfield{author}{\bibinfo{person}{Kassem Fawaz} {and} \bibinfo{person}{Kang~G.
  Shin}.} \bibinfo{year}{2014}\natexlab{}.
\newblock \showarticletitle{Location Privacy Protection for Smartphone Users}.
  In \bibinfo{booktitle}{{\em {ACM} {SIGSAC} Conference on Computer and
  Communications Security}}, \bibfield{editor}{\bibinfo{person}{Gail{-}Joon
  Ahn}, \bibinfo{person}{Moti Yung}, {and} \bibinfo{person}{Ninghui Li}}
  (Eds.). \bibinfo{publisher}{{ACM}}, \bibinfo{pages}{239--250}.
\newblock


\bibitem[\protect\citeauthoryear{Freudiger, Shokri, and Hubaux}{Freudiger
  et~al\mbox{.}}{2012}]%
        {FreudigerSH12}
\bibfield{author}{\bibinfo{person}{Julien Freudiger}, \bibinfo{person}{Reza
  Shokri}, {and} \bibinfo{person}{Jean-Pierre Hubaux}.}
  \bibinfo{year}{2012}\natexlab{}.
\newblock \showarticletitle{Evaluating the privacy risk of location-based
  services}.
\newblock In \bibinfo{booktitle}{{\em Financial Cryptography and Data
  Security}}. \bibinfo{publisher}{Springer}, \bibinfo{pages}{31--46}.
\newblock


\bibitem[\protect\citeauthoryear{Gambs, Killijian, and del Prado~Cortez}{Gambs
  et~al\mbox{.}}{2011}]%
        {GambsKC11}
\bibfield{author}{\bibinfo{person}{S{\'{e}}bastien Gambs},
  \bibinfo{person}{Marc{-}Olivier Killijian}, {and}
  \bibinfo{person}{Miguel~N{\'{u}}{\~{n}}ez del Prado~Cortez}.}
  \bibinfo{year}{2011}\natexlab{}.
\newblock \showarticletitle{Show Me How You Move and {I} Will Tell You Who You
  Are}.
\newblock \bibinfo{journal}{{\em Transactions on Data Privacy\/}}
  \bibinfo{volume}{4}, \bibinfo{number}{2} (\bibinfo{year}{2011}),
  \bibinfo{pages}{103--126}.
\newblock


\bibitem[\protect\citeauthoryear{Gedik and Liu}{Gedik and Liu}{2005}]%
        {GedikL05}
\bibfield{author}{\bibinfo{person}{Bugra Gedik} {and} \bibinfo{person}{Ling
  Liu}.} \bibinfo{year}{2005}\natexlab{}.
\newblock \showarticletitle{Location Privacy in Mobile Systems: {A}
  Personalized Anonymization Model}. In \bibinfo{booktitle}{{\em 25th
  International Conference on Distributed Computing Systems {(ICDCS}}}.
  \bibinfo{publisher}{{IEEE} Computer Society}, \bibinfo{pages}{620--629}.
\newblock


\bibitem[\protect\citeauthoryear{Golle and Partridge}{Golle and
  Partridge}{2009}]%
        {GolleP09}
\bibfield{author}{\bibinfo{person}{Philippe Golle} {and} \bibinfo{person}{Kurt
  Partridge}.} \bibinfo{year}{2009}\natexlab{}.
\newblock \showarticletitle{On the Anonymity of Home/Work Location Pairs}. In
  \bibinfo{booktitle}{{\em International Conference on Pervasive Computing}}
  {\em (\bibinfo{series}{LNCS})}, \bibfield{editor}{\bibinfo{person}{Hideyuki
  Tokuda}, \bibinfo{person}{Michael Beigl}, \bibinfo{person}{Adrian Friday},
  \bibinfo{person}{A.~J.~Bernheim Brush}, {and} \bibinfo{person}{Yoshito Tobe}}
  (Eds.), Vol.~\bibinfo{volume}{5538}. \bibinfo{publisher}{Springer},
  \bibinfo{pages}{390--397}.
\newblock


\bibitem[\protect\citeauthoryear{Gruteser and Grunwald}{Gruteser and
  Grunwald}{2003}]%
        {GruteserG03}
\bibfield{author}{\bibinfo{person}{Marco Gruteser} {and} \bibinfo{person}{Dirk
  Grunwald}.} \bibinfo{year}{2003}\natexlab{}.
\newblock \showarticletitle{Anonymous Usage of Location-Based Services Through
  Spatial and Temporal Cloaking}. In \bibinfo{booktitle}{{\em International
  conference on Mobile systems, applications and services}}.
  \bibinfo{publisher}{ACM}, \bibinfo{pages}{31--42}.
\newblock


\bibitem[\protect\citeauthoryear{Hoh and Gruteser}{Hoh and Gruteser}{2005}]%
        {HohG05}
\bibfield{author}{\bibinfo{person}{B. Hoh} {and} \bibinfo{person}{M.
  Gruteser}.} \bibinfo{year}{2005}\natexlab{}.
\newblock \showarticletitle{Protecting Location Privacy Through Path
  Confusion}. In \bibinfo{booktitle}{{\em International Conference on Security
  and Privacy for Emerging Areas in Communications Networks}}.
  \bibinfo{pages}{194--205}.
\newblock
\showDOI{%
\url{https://doi.org/10.1109/SECURECOMM.2005.33}}


\bibitem[\protect\citeauthoryear{Kido, Yanagisawa, and Satoh}{Kido
  et~al\mbox{.}}{2005}]%
        {KidoYS05}
\bibfield{author}{\bibinfo{person}{H. Kido}, \bibinfo{person}{Y. Yanagisawa},
  {and} \bibinfo{person}{T. Satoh}.} \bibinfo{year}{2005}\natexlab{}.
\newblock \showarticletitle{An anonymous communication technique using dummies
  for location-based services}. In \bibinfo{booktitle}{{\em Pervasive Services,
  2005. ICPS '05. Proceedings. International Conference on}}.
  \bibinfo{pages}{88--97}.
\newblock


\bibitem[\protect\citeauthoryear{Krumm}{Krumm}{2007}]%
        {Krumm07}
\bibfield{author}{\bibinfo{person}{John Krumm}.}
  \bibinfo{year}{2007}\natexlab{}.
\newblock \showarticletitle{Inference Attacks on Location Tracks}. In
  \bibinfo{booktitle}{{\em 5th International Conference on Pervasive
  Computing}} {\em (\bibinfo{series}{LNCS})},
  \bibfield{editor}{\bibinfo{person}{Anthony LaMarca}, \bibinfo{person}{Marc
  Langheinrich}, {and} \bibinfo{person}{Khai~N. Truong}} (Eds.),
  Vol.~\bibinfo{volume}{4480}. \bibinfo{publisher}{Springer},
  \bibinfo{pages}{127--143}.
\newblock


\bibitem[\protect\citeauthoryear{Lu, Jensen, and Yiu}{Lu et~al\mbox{.}}{2008}]%
        {LuJY08}
\bibfield{author}{\bibinfo{person}{Hua Lu}, \bibinfo{person}{Christian~S.
  Jensen}, {and} \bibinfo{person}{Man~Lung Yiu}.}
  \bibinfo{year}{2008}\natexlab{}.
\newblock \showarticletitle{PAD: privacy-area aware, dummy-based location
  privacy in mobile services}. In \bibinfo{booktitle}{{\em ACM International
  Workshop on Data Engineering for Wireless and Mobile Access}}.
  \bibinfo{publisher}{ACM}, \bibinfo{pages}{16--23}.
\newblock
\showISBNx{978-1-60558-221-4}
\showDOI{%
\url{https://doi.org/10.1145/1626536.1626540}}


\bibitem[\protect\citeauthoryear{Ma and Chen}{Ma and Chen}{2014}]%
        {Ma2014}
\bibfield{author}{\bibinfo{person}{Changsha Ma} {and}
  \bibinfo{person}{Chang~Wen Chen}.} \bibinfo{year}{2014}\natexlab{}.
\newblock \showarticletitle{Nearby Friend Discovery with
  Geo-indistinguishability to Stalkers}.
\newblock \bibinfo{journal}{{\em Procedia Computer Science\/}}
  \bibinfo{volume}{34} (\bibinfo{year}{2014}), \bibinfo{pages}{352--359}.
\newblock


\bibitem[\protect\citeauthoryear{Meyerowitz and Choudhury}{Meyerowitz and
  Choudhury}{2009}]%
        {MeyerowitzC09}
\bibfield{author}{\bibinfo{person}{Joseph~T. Meyerowitz} {and}
  \bibinfo{person}{Romit~Roy Choudhury}.} \bibinfo{year}{2009}\natexlab{}.
\newblock \showarticletitle{Hiding stars with fireworks: location privacy
  through camouflage}. In \bibinfo{booktitle}{{\em 15th Annual International
  Conference on Mobile Computing and Networking (MOBICOM)}},
  \bibfield{editor}{\bibinfo{person}{Kang~G. Shin}, \bibinfo{person}{Yongguang
  Zhang}, \bibinfo{person}{Rajive Bagrodia}, {and} \bibinfo{person}{Ramesh
  Govindan}} (Eds.). \bibinfo{publisher}{ACM}, \bibinfo{pages}{345--356}.
\newblock


\bibitem[\protect\citeauthoryear{Shokri}{Shokri}{2015}]%
        {Shokri15}
\bibfield{author}{\bibinfo{person}{Reza Shokri}.}
  \bibinfo{year}{2015}\natexlab{}.
\newblock \showarticletitle{Privacy Games: Optimal User-Centric Data
  Obfuscation}.
\newblock \bibinfo{journal}{{\em PoPETs\/}} \bibinfo{volume}{2015},
  \bibinfo{number}{2} (\bibinfo{year}{2015}), \bibinfo{pages}{299--315}.
\newblock


\bibitem[\protect\citeauthoryear{Shokri, Freudiger, Jadliwala, and
  Hubaux}{Shokri et~al\mbox{.}}{2009}]%
        {ShokriFJH09}
\bibfield{author}{\bibinfo{person}{Reza Shokri}, \bibinfo{person}{Julien
  Freudiger}, \bibinfo{person}{Murtuza Jadliwala}, {and}
  \bibinfo{person}{Jean{-}Pierre Hubaux}.} \bibinfo{year}{2009}\natexlab{}.
\newblock \showarticletitle{A distortion-based metric for location privacy}. In
  \bibinfo{booktitle}{{\em {ACM} Workshop on Privacy in the Electronic Society,
  {WPES}}}, \bibfield{editor}{\bibinfo{person}{Ehab Al{-}Shaer} {and}
  \bibinfo{person}{Stefano Paraboschi}} (Eds.). \bibinfo{publisher}{{ACM}},
  \bibinfo{pages}{21--30}.
\newblock


\bibitem[\protect\citeauthoryear{Shokri, Theodorakopoulos, Le~Boudec, and
  Hubaux}{Shokri et~al\mbox{.}}{2011}]%
        {Quant2011SP}
\bibfield{author}{\bibinfo{person}{Reza Shokri}, \bibinfo{person}{George
  Theodorakopoulos}, \bibinfo{person}{Jean-Yves Le~Boudec}, {and}
  \bibinfo{person}{Jean-Pierre Hubaux}.} \bibinfo{year}{2011}\natexlab{}.
\newblock \showarticletitle{Quantifying location privacy}. In
  \bibinfo{booktitle}{{\em Security and privacy (sp), 2011 ieee symposium on}}.
  IEEE, \bibinfo{pages}{247--262}.
\newblock


\bibitem[\protect\citeauthoryear{Shokri, Theodorakopoulos, Troncoso, Hubaux,
  and Le~Boudec}{Shokri et~al\mbox{.}}{2012}]%
        {Opt2012CCS}
\bibfield{author}{\bibinfo{person}{Reza Shokri}, \bibinfo{person}{George
  Theodorakopoulos}, \bibinfo{person}{Carmela Troncoso},
  \bibinfo{person}{Jean-Pierre Hubaux}, {and} \bibinfo{person}{Jean-Yves
  Le~Boudec}.} \bibinfo{year}{2012}\natexlab{}.
\newblock \showarticletitle{Protecting location privacy: optimal strategy
  against localization attacks}. In \bibinfo{booktitle}{{\em Proceedings of the
  2012 ACM conference on Computer and communications security}}. ACM,
  \bibinfo{pages}{617--627}.
\newblock


\bibitem[\protect\citeauthoryear{Wang, Xu, He, Zhang, Li, and Xu}{Wang
  et~al\mbox{.}}{2012}]%
        {WangXHZLX12}
\bibfield{author}{\bibinfo{person}{Yu Wang}, \bibinfo{person}{Dingbang Xu},
  \bibinfo{person}{Xiao He}, \bibinfo{person}{Chao Zhang}, \bibinfo{person}{Fan
  Li}, {and} \bibinfo{person}{Bin Xu}.} \bibinfo{year}{2012}\natexlab{}.
\newblock \showarticletitle{L2P2: Location-aware location privacy protection
  for location-based services}. In \bibinfo{booktitle}{{\em INFOCOM, 2012
  Proceedings IEEE}}. \bibinfo{pages}{1996--2004}.
\newblock
\showISSN{0743-166X}
\showDOI{%
\url{https://doi.org/10.1109/INFCOM.2012.6195577}}


\bibitem[\protect\citeauthoryear{You, Peng, and Lee}{You et~al\mbox{.}}{2007}]%
        {YouPL07}
\bibfield{author}{\bibinfo{person}{Tun-Hao You}, \bibinfo{person}{Wen-Chih
  Peng}, {and} \bibinfo{person}{Wang-Chien Lee}.}
  \bibinfo{year}{2007}\natexlab{}.
\newblock \showarticletitle{Protecting Moving Trajectories with Dummies}. In
  \bibinfo{booktitle}{{\em International Conference on Mobile Data
  Management}}. \bibinfo{pages}{278--282}.
\newblock


\bibitem[\protect\citeauthoryear{Zheng, Zhang, Xie, and Ma}{Zheng
  et~al\mbox{.}}{2009}]%
        {ZhengZXM09}
\bibfield{author}{\bibinfo{person}{Yu Zheng}, \bibinfo{person}{Lizhu Zhang},
  \bibinfo{person}{Xing Xie}, {and} \bibinfo{person}{Wei-Ying Ma}.}
  \bibinfo{year}{2009}\natexlab{}.
\newblock \showarticletitle{Mining Interesting Locations and Travel Sequences
  from GPS Trajectories}. In \bibinfo{booktitle}{{\em Proceedings of the 18th
  International Conference on World Wide Web}}. \bibinfo{publisher}{ACM}, 10.
\newblock


\end{thebibliography}

\end{document}